\documentclass[lettersize,journal]{IEEEtran}
\usepackage{amsmath,amsfonts,amssymb}
\usepackage{amsthm}
\usepackage{algorithmic}
\usepackage{multirow}
\usepackage{booktabs}
\usepackage[ruled,linesnumbered]{algorithm2e}
\usepackage{array}
\usepackage[caption=false,font=footnotesize,labelfont=rm,textfont=rm]{subfig}
\usepackage{textcomp}
\usepackage{stfloats}
\usepackage{url}
\usepackage{verbatim}
\usepackage{graphicx}
\usepackage[colorlinks,citecolor=blue,linkcolor=blue,urlcolor=blue]{hyperref}
\usepackage{cite}
\usepackage{xcolor}
\usepackage{tikz}
\usetikzlibrary{shapes.geometric,arrows.meta,positioning,calc,decorations.pathreplacing}
\newcommand{\bm}[1]{\boldsymbol{#1}}
\newtheorem{theorem}{Theorem}
\newtheorem{lemma}{Lemma}
\newtheorem{proposition}{Proposition}

\newtheorem{remark}{Remark}
\newtheorem{definition}{Definition}

\begin{document}

\title{Sparse Fluid Antenna Arrays: Continuous Position Design Beyond Classical DOF Limits}
\author{Tuo~Wu, Jie Tang,  Ye Tian, Cheng Zeng, Matthew  C. Valenti, \emph{Fellow}, \emph{IEEE},  Hing Cheung So, \emph{Fellow, IEEE}
\thanks{(\textit{Corresponding author: Jie Tang.})}
	\thanks{T. Wu, and J. Tang are with the School of Electronic and Information Engineering, South China University of Technology, Guangzhou 510640, China (E-mail: $\rm \{wutuo,eejtang\}@scut.edu.cn$). Y. Tian   with the Faculty of Electrical Engineering and Computer Science, Ningbo University, Ningbo 315211, China (E-mail: $\rm tianye1@nbu.edu.cn$). C. Zeng is with the School of Electronic and Optical Engineering, Nanjing University of Science and Technology, Nanjing 210094, China, 	(E-mail: $\rm czeng@njust.edu.cn$). M. C. Valenti is with the Lane Department of Computer Science and Electrical Engineering, West Virginia University, Morgantown, USA (E-mail: $\rm valenti@ieee.org$). H. C. So is with the Department of Electrical Engineering, City University of Hong Kong, Hong Kong, China (E-mail:   $\rm hcso@ee.cityu.edu.hk$). }}

	\markboth{IEEE TRANSACTIONS ON SIGNAL PROCESSING~Vol.~XX, No.~XX, XX~2026}{}%

\maketitle

\begin{abstract}
Classical sparse arrays, including the nested, coprime, and minimum redundancy arrays (MRAs), achieve $O(N^2)$ degrees-of-freedom (DOFs) by exploiting non-uniform spacing on the half-wavelength grid. Yet this grid constraint fundamentally caps the achievable DOFs and Cram\'{e}r--Rao bound (CRB) at levels dictated solely by the antenna count $N$, regardless of how large the deployment region is. Fluid antenna system (FAS), which continuously repositions a single physical element across a deployment region $[0, D]$, breaks this limit by freeing antenna positions from the discrete grid entirely. This paper establishes the theoretical foundations of sparse FAS design for direction-of-arrival (DOA) estimation and shows that continuous position freedom unlocks three compounding advantages over the classical designs. \emph{First}, we derive a universal dual DOF bound and prove that FAS-optimized positions can approach it, growing the DOF linearly with $D/\lambda$ , where $\lambda$ is the signal wavelength, rather than saturating at $O(N^2)$. \emph{Second}, the CRB scales as $O(1/D^{2L})$ for $L$ sources, a $(D/(N^2 d_0))^{2L}$ improvement over the best grid design, with $d_0 = \lambda/2$ and D-optimal positions admitting closed-form solution for single sources and efficient Frank-Wolfe algorithm for multiple sources. \emph{Third}, we propose a two-stage FAS-MUSIC approach that combines coarray MUSIC disambiguation with full-aperture local maximum likelihood (ML) refinement to track the CRB, overcoming the grating-lobe ambiguity inherent in large-aperture non-uniform arrays. Robustness to minimum spacing constraints, mutual coupling, and finite position accuracy is also analyzed. Extensive simulations show that FAS-MUSIC achieves $17.5\times$ lower root mean squared error (RMSE) than uniform linear array (ULA) MUSIC and that FAS with $4$ antennas outperforms MRA with $8$ antennas, gains that are unattainable by any grid-constrained design.
\end{abstract}


\section{Introduction}\label{sec:intro}

\IEEEPARstart{S}{parse} array design has been a cornerstone of array signal processing for decades, aiming to maximize the spatial degrees-of-freedom (DOFs) achievable with a limited number of physical sensors~\cite{moffet_mra,pal_nested,vaidyanathan_coprime}. The key insight is that non-uniform inter-element spacing generates a virtual array, the \emph{difference coarray}, whose aperture far exceeds the physical array, enabling subspace methods such as MUSIC~\cite{stoica_music} to resolve far more sources than physical elements. Despite this ingenuity, all classical designs share an inescapable bottleneck: \emph{the achievable DOF is bounded by $O(N^2)$ and the Cram\'{e}r--Rao bound (CRB) scales no better than $O(1/(N^2 d_0)^{2L})$, where L denotes the source number, both quantities are dictated solely by the antenna count $N$ and the half-wavelength inter-element spacing $d_0= \lambda/2$ with $\lambda$ being the wavelength.} Enlarging the deployment region does not help, spreading grid elements further only creates coarray holes that eliminate the DOF advantage entirely.

Three families of sparse arrays have dominated the literature. Minimum redundancy arrays (MRAs)~\cite{moffet_mra} achieve the longest contiguous coarray for a given $N$ but require exhaustive combinatorial search, which becomes intractable for large $N$. Nested arrays~\cite{pal_nested} provide a systematic two-level design with $O(N^2)$ contiguous coarray elements and closed-form position expressions, at the cost of a dense inner subarray that is susceptible to mutual coupling. Coprime arrays~\cite{vaidyanathan_coprime} mitigate coupling by ensuring all spacings are at least one wavelength, but their coarray contains holes that reduce the usable DOFs. Subsequent works have proposed numerous variants such as generalized coprime arrays~\cite{qin_generalized_coprime}, super-nested arrays~\cite{liu_super_nested}, and augmented coprime arrays, each improving specific aspects but all sharing a fundamental constraint: \emph{antenna positions are restricted to integer multiples of the half-wavelength spacing $d_0 = \lambda/2$}.

Fluid antenna system (FAS)~\cite{wong_fas,wang2021fluid} fundamentally changes this landscape. FAS is a hardware-agnostic concept that treats the antenna as a reconfigurable physical-layer resource, emphasizing position and shape flexibility~\cite{new2025tut,new2025flar,wu2024flu,new_fas_mimo}. Practical FAS implementations span diverse technologies including movable elements~\cite{Zhu-Wong-2024}, liquid-based antennas~\cite{shen2024design}, and reconfigurable pixel antennas~\cite{zhang2024pixel,tong-2025pixel,tong2025design}. By enabling antenna positions to be continuously reconfigured within a deployment region $[0, D]$, FAS removes the half-wavelength grid constraint and opens a strictly larger design space for array optimization. This continuous position freedom raises a natural and important question: \emph{what are the fundamental DOF limits when antenna positions are freed from the grid, and how should positions be optimally designed to exploit this enlarged design space?}

To appreciate why this position freedom is transformative, consider a concrete scenario: $N = 6$ antennas deployed over a $D = 40d_0$ region. A grid-based array selects 6 positions from 81 available grid points $\{0, d_0, \ldots, 40d_0\}$. The fundamental bottleneck is combinatorial: 6 elements produce at most $\binom{6}{2} = 15$ distinct positive pairwise differences. For the difference coarray to be contiguous up to lag $M_c$, all integers $\{1, 2, \ldots, M_c\}$ must appear as differences, so the contiguous DOF is capped at $2 \times 15 + 1 = 31$, regardless of how large the deployment region is. In practice, the best-known MRA for $N = 6$ achieves DOF $= 19$ with aperture $20d_0$, only half the available deployment region. Spreading grid elements further across $40d_0$ would not help; the 15 available differences cannot fill the 40 required integer lags, and the coarray would develop large holes that eliminate the DOF entirely. The aperture of grid-based arrays is thus \emph{locked} to the antenna count, not to the deployment region.

FAS removes this lock. Because fluid antenna positions are real-valued and continuously adjustable, one can simply slide the antennas to the two endpoints of the deployment region, placing, say, three elements near $p = 0$ and three near $p = D = 40d_0$ \footnote{This endpoint clustering example is intended only to illustrate aperture-driven scaling; practical designs, accounting for spacing and coupling constraints, distribute elements across the aperture as developed later in the paper.}. This instantly achieves two compounding advantages. \emph{First}, the physical aperture becomes $D = 40d_0$ regardless of $N$, yielding angular resolution $\delta\theta \approx \lambda/D$, which is twice finer than the MRA's $\lambda/(20d_0)$ with the same 6 antennas. \emph{Second}, the position variance $\mu_2 = \frac{1}{N}\sum_n (p_n - \bar{p})^2$ directly controls the Cram\'{e}r--Rao bound (CRB) on direction-of-arrival (DOA) estimation accuracy. With all elements concentrated near the two endpoints, $\mu_2 = (40d_0)^2/4 = 400d_0^2$, compared to $\mu_2 \leq (20d_0)^2/4 = 100d_0^2$ for the MRA, a fourfold reduction in CRB using identical hardware.

The large aperture $D$ achievable by FAS does raise a known problem: the MUSIC spectrum of large-aperture non-uniform arrays can exhibit strong \emph{grating lobes}, creating spurious estimation peaks. Crucially, however, the continuous position freedom that creates this challenge also resolves it. Because FAS positions are real-valued and fully adjustable, they can be designed so that the resulting difference coarray contains a long \emph{contiguous} virtual uniform linear array (ULA) segment with half-wavelength spacing, which is by construction free of grating-lobe ambiguities, while the full physical aperture $D$ simultaneously drives the CRB to its minimum. The two quantities, DOF and accuracy, are thus not in tension for FAS; they are jointly optimized through position design.

Taken together, sparse FAS offers three compounding advantages over classical grid-based designs:
\begin{enumerate}
    \item \textbf{Aperture-scaled DOF}: the achievable DOF grows with $D/\lambda$ rather than $N^2$, strictly surpassing every grid design when $D$ is large;
    \item \textbf{Aperture-scaled CRB}: the CRB scales as $1/D^{2L}$ rather than $1/(N^2 d_0)^{2L}$, a CRB reduction that grows unboundedly with the deployment region;
    \item \textbf{Grating-lobe-free estimation by design}: the coarray structure ensures an unambiguous virtual aperture, so the full physical aperture is exploitable for precision, not just for coarse direction finding.
\end{enumerate}
A fourth practical benefit is tractability: unlike the NP-hard MRA combinatorial search on the grid, the FAS position design is a smooth continuous problem over $[0,D]^N$ that admits efficient polynomial-time optimization.

Position reconfigurability in FAS provides a new way to exploit the spatial domain, unlike conventional fixed-position antennas (FPAs), and has attracted rapidly expanding research interest. Fundamental channel characterization efforts have established spatial correlation models~\cite{FAS22,ramirez2024new,new2025channel} and analyzed diversity performance~\cite{new2024fluid,new2024an}. The combination of multiple-input multiple-output (MIMO) strategy and FAS leads to joint antenna position optimization and beamforming problems~\cite{cite10416896,LZhu23}. FAS has been applied to integrated sensing and communications (ISAC)~\cite{Zhou-isac2024}, non-orthogonal multiple access (NOMA)~\cite{TWuTCOMM26}, reconfigurable intelligent surface (RIS)-aided systems~\cite{TWuTVT25}, and secrecy communication~\cite{TuoW}. Fluid antenna multiple access schemes~\cite{H4_wong2022FAMA,H5_wong2023fast,H6_wong2023sFAMA} exploit position reconfigurability for interference-free multiuser communication.

While position optimization for FAS has been extensively studied in the communication community, DOA estimation, where the difference coarray structure plays a central role, remains largely unexplored. In particular, no prior work has addressed the question of how to design FAS positions to maximize the coarray DOF or minimize the CRB for DOA estimation, which is fundamentally different from communication-oriented design. Recent exceptions include~\cite{HXuTWC25} and~\cite{TWuJSTSP25}, which demonstrated the potential of sparse FAS for DOA estimation and proposed FAS as a new paradigm for array signal processing, respectively. However, neither provides the complete theoretical framework, including the dual DOF bound, D-optimal position design, and a two-stage estimation algorithm with near-CRB performance guarantees, that this paper establishes.

Several related concepts in the literature bear surface similarity to  sparse FAS design, but are fundamentally different in key respects. \emph{Synthetic aperture radar/sonar (SAR/SAS)} also creates a large virtual aperture by moving a sensor along a trajectory~\cite{cutrona_sar}, but the trajectory and its parameters are determined by the platform motion and cannot be jointly optimized with the signal processing; coherent processing requires phase stability over the entire trajectory, demanding stringent motion compensation. \emph{Sparse-recovery-based array interpolation}~\cite{abramovich_interpolation} estimates the virtual coarray signal via compressed sensing, but assumes a fixed array geometry and does not optimize the physical positions. \emph{Reconfigurable/switched arrays}~\cite{haupt_reconfigurable} select a subset of elements from a fixed dense array, operating in a discrete feasible set equivalent to $\mathcal{P}_{\mathrm{grid}}$, the NP-hard MRA problem remains. In contrast, FAS allows \emph{continuous} and \emph{simultaneous} optimization of all element positions within $[0,D]$, enabling tractable gradient-based design and runtime adaptability that none of these alternatives provide.

Against this backdrop, this paper establishes the theoretical foundations of sparse FAS design for DOA estimation. The main contributions are summarized as follows:

\begin{itemize}
    \item \textbf{\textit{Fundamental DOF Limit:}} We establish that the coarray DOF of any $N$-element array in $[0,D]$ is universally bounded by $\min(N^2-N+1,\, 2\lfloor D/d_0\rfloor+1)$, a dual bound combining a combinatorial term and a geometric term. For grid-based arrays, coarray holes caused by insufficient pairwise differences prevent reaching this bound when $D \gg O(N^2 d_0)$. By contrast, FAS continuously optimized positions can remain near this bound; we characterize the achievable FAS DOF gap via the Frank-Wolfe optimality condition and show it to be small in the practically relevant large-aperture regime.

    \item \textbf{\textit{Optimal Position Design:}} We formulate the FAS position optimization as a continuous problem over the difference coarray. For the single-source case, a closed-form globally optimal design is derived. For the multi-source scenarios, we develop a Frank-Wolfe algorithm over continuous measures with a certifiable optimality gap.

    \item \textbf{\textit{Practical Robust Design:}} We incorporate minimum spacing constraints ($|p_i - p_j| \geq d_{\min}$), mutual coupling effects, and finite position accuracy into the optimization framework, and quantify the DOF loss due to each imperfection.

    \item \textbf{\textit{Two-Stage FAS-MUSIC Algorithm:}} We propose a two-stage DOA estimator tailored for sparse FAS. Stage~1 exploits the difference coarray via spatial smoothing MUSIC to obtain unambiguous coarse estimates, resolving the grating-lobe ambiguity inherent in large-aperture non-uniform arrays. Stage~2 refines these estimates via local maximum likelihood (ML) estimation using the full FAS array, attaining near-CRB accuracy. An adaptive variant that jointly optimizes positions and estimates is also developed.

    \item \textbf{\textit{Comprehensive Evaluation:}} Extensive simulations compare sparse FAS against ULA, nested, coprime, and MRA designs across DOF, CRB, root mean square error (RMSE), and super-resolution metrics, demonstrating consistent and significant advantages of both the FAS design and the proposed estimator.
\end{itemize}

The remainder is organized as follows. Section~\ref{sec:model} reviews the signal model and coarray framework. Section~\ref{sec:dof} derives the DOF limits. Section~\ref{sec:design} develops the optimal position design algorithms and the two-stage FAS-MUSIC estimator. Section~\ref{sec:robust} addresses practical constraints. Section~\ref{sec:sim} provides simulation results, and finally, conclusions are drawn in Section~\ref{sec:conc}.

\emph{Notations}: Standard notation is used throughout. $\mathbf{A}$, $\mathbf{a}$, and $a$ denote matrix, vector, and scalar, respectively. $(\cdot)^T$, $(\cdot)^H$, $(\cdot)^{-1}$ denote transpose, conjugate transpose, and inverse. $\odot$ denotes the Khatri-Rao (column-wise Kronecker) product; the same symbol is used for the Hadamard (element-wise) product when both operands share the same dimensions. $\mathrm{vec}(\cdot)$ vectorizes a matrix.

\section{Signal Model and Coarray Framework}\label{sec:model}

This section builds the theoretical framework for FAS-based DOA estimation. We first review classical sparse arrays and their grid-imposed constraints, then introduce FAS and explain how its continuous position freedom overcomes those limitations, and finally present the formal FAS signal model and coarray formulation used throughout the paper.

\subsection{Classical Sparse Arrays and their Limitations}

Consider $N$ antenna elements with positions $\{p_n\}_{n=1}^N$ within a deployment region of aperture $D$, receiving $L$ narrowband far-field sources at angle vector $\bm{\theta}$. The received signal follows the standard model $\mathbf{x}(t) = \mathbf{A}(\bm{\theta})\mathbf{s}(t) + \mathbf{n}(t)$, with spatial covariance $\mathbf{R} = \mathbf{A}\mathbf{R}_s\mathbf{A}^H + \sigma^2\mathbf{I}$. A key property of $\mathbf{R}$ is that its vectorization $\mathbf{z} = \mathrm{vec}(\mathbf{R})$ is equivalent to measurements of a virtual array at the \emph{difference coarray} positions $\mathbb{D} = \{p_i - p_j : 1 \leq i,j \leq N\}$, and the number of resolvable sources, i.e., DOF,  equals the cardinality of the longest contiguous segment of $\mathbb{D}$ on the Nyquist grid.

Classical sparse arrays, nested arrays~\cite{pal_nested}, coprime arrays~\cite{vaidyanathan_coprime}, and MRAs~\cite{moffet_mra}, exploit this coarray structure by carefully choosing positions on the half-wavelength grid $\{0, d_0, 2d_0, \ldots\}$:
\begin{itemize}
    \item \emph{Nested Array} achieves $O(N^2)$ contiguous coarray elements with a closed-form two-level design.
    \item \emph{Coprime Array} reduces mutual coupling but introduces coarray holes that lower the usable DOF.
    \item \emph{MRA} maximizes the contiguous DOF for a given $N$, but finding optimal positions is NP-hard.
\end{itemize}
Fig.~\ref{fig:coarray_vis} illustrates these designs for $N = 6$.

\begin{figure}[!t]
    \centering
    \includegraphics[width=0.95\columnwidth]{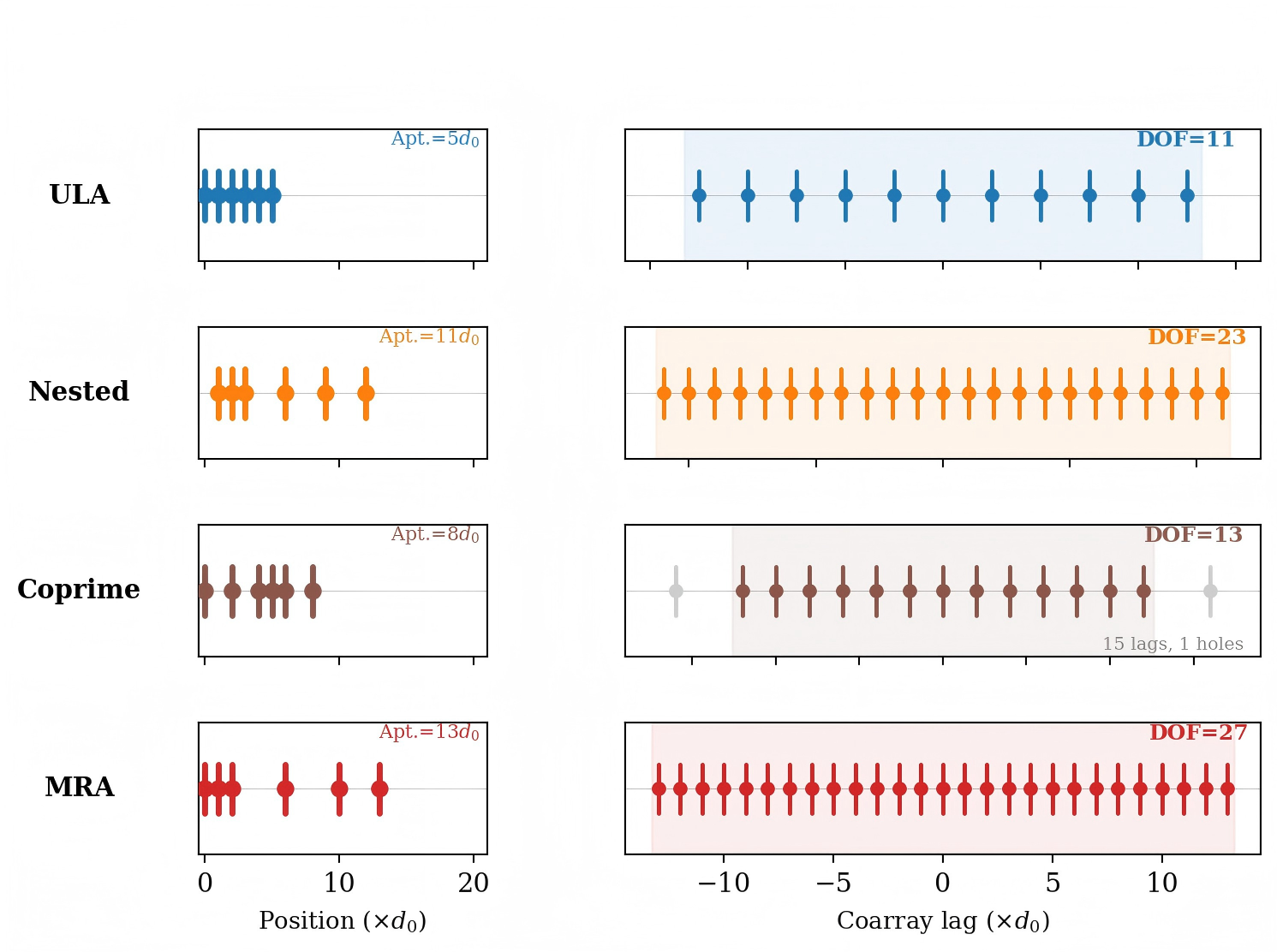}
    \caption{Physical array positions (left) and their difference coarrays (right) for four classical array designs with $N=6$. Colored markers denote contiguous coarray elements; gray markers indicate non-contiguous lags (holes). The contiguous DOF and total number of lags are annotated for each design.}
    \label{fig:coarray_vis}
\end{figure}

Despite their sophistication, all three designs share a fundamental bottleneck: \emph{positions are restricted to integer multiples of $d_0 = \lambda/2$}. This grid constraint creates two coupled limitations. First, the achievable aperture is locked to $O(N^2 d_0)$ regardless of the physical deployment region $D$; for $N = 6$, the best MRA achieves only $20d_0$ even if $D = 40d_0$ or larger. Second, with only $N(N-1)/2$ distinct positive pairwise differences available, the contiguous DOF is bounded by $N^2-N+1$, scaling with $N^2$, not with $D/d_0$. Extending the physical aperture by spreading grid elements further actually \emph{degrades} performance: the sparse differences cannot fill the required integer lags, creating coarray holes and eliminating the DOF advantage entirely. The aperture of grid-based arrays is thus fundamentally coupled to the antenna count, not to the deployment region.

\subsection{FAS: Overcoming Grid Limitations}\label{subsec:why_fas}

FAS~\cite{wong_fas,new_fas_mimo} removes the grid restriction by allowing antenna positions to be continuously reconfigured within $[0, D]$. This seemingly simple relaxation has profound consequences.

\subsubsection{Aperture and DOF Advantages}
With real-valued positions, FAS can place antennas at any location in $[0, D]$. By positioning elements near the two endpoints ($p_1 \approx 0$, $p_N \approx D$), the physical aperture immediately reaches $D$, far beyond the $O(N^2 d_0)$ ceiling of any grid-based design. For $N = 6$ and $D = 40d_0$, FAS achieves aperture $40d_0$ versus\ the MRA's $20d_0$, giving twofold finer angular resolution $\delta\theta \approx \lambda/D$. Furthermore, the position variance $\mu_2 = \frac{1}{N}\sum_n(p_n - \bar{p})^2$ directly controls the CRB; with elements clustered near the two endpoints, $\mu_2 = (40d_0)^2/4 = 400d_0^2$ versus $\mu_2 \leq 100d_0^2$ for the MRA, a fourfold reduction in CRB with identical hardware.

Moreover, the DOF ceiling of $N^2-N+1$ for grid arrays is no longer binding: FAS positions are flexible and can be designed so that the difference coarray $\mathbb{D}$ contains a long contiguous virtual ULA segment of length $2\lfloor D/d_0 \rfloor + 1$, growing linearly with the deployment region rather than with $N^2$.

\subsubsection{Grating Lobes the Coarray Solution}
The large aperture of FAS brings a well-known challenge: the MUSIC spectrum of large-aperture non-uniform arrays exhibits strong \emph{grating lobes}, making DOA estimation unreliable with standard algorithms. Critically, FAS's continuous position freedom also provides the solution: positions are designed so that the difference coarray $\mathbb{D}$ includes a \emph{contiguous} virtual ULA segment $\mathbb{D}_c$ with half-wavelength spacing. This segment is by construction free of grating-lobe ambiguities, while the full physical aperture $D$ governs the CRB. The dual objective of contiguous coarray coverage and large physical aperture is simultaneously achievable through continuous position optimization, a combination that is impossible for grid-based designs.

\subsubsection{Tractable Position Design}
On the grid, finding optimal positions for MRA is NP-hard. With continuous positions over $[0, D]^N$, the position design becomes a smooth problem solvable in polynomial time via gradient-based optimization. 

\subsection{FAS Signal Model}\label{subsec:fas_model}

Having identified the FAS advantages, we now formalize the signal model. A single fluid antenna element visits $N$ pre-designed locations $\{p_n\}_{n=1}^N \subset [0, D]$, following a sequential protocol: the antenna is first configured at $p_1$, collects $N_p$ temporal snapshots, then repositions to $p_2$ for another $N_p$ snapshots, and so on until $p_N$. Sources are assumed stationary (constant DOA vector $\bm{\theta}$) throughout the traversal, requiring total observation time $T_\text{obs} = N N_p T_s \ll T_c$, where $T_c$ is the signal coherence time. As a concrete example, for $N = 6$ and $D = 40d_0$, the optimized positions are $\{0,\, 3d_0,\, 8d_0,\, 32d_0,\, 37d_0,\, 40d_0\}$; with $N_p = 500$ snapshots per position and $T_s = 1\,\mu\text{s}$, the total observation time is $T_\text{obs} = 3\,\text{ms}$, well within typical quasi-static coherence times of tens to hundreds of milliseconds.

Stacking measurements across all $N$ positions yields the data matrix $\mathbf{X} \in \mathbb{C}^{N \times N_p}$, where row $n$ contains the $N_p$ snapshots at position $p_n$. Under the stationarity assumption, this is equivalent to an $N$-element virtual array receiving the signal:
\begin{equation}\label{eq:signal}
    \mathbf{x}(t) = \mathbf{A}(\bm{\theta}) \mathbf{s}(t) + \mathbf{n}(t),
\end{equation}
where $\mathbf{A}(\bm{\theta}) = [\mathbf{a}(\theta_1), \ldots, \mathbf{a}(\theta_L)] \in \mathbb{C}^{N \times L}$ is the steering matrix with
\begin{equation}\label{eq:steering}
    [\mathbf{a}(\theta_\ell)]_n = e^{j \frac{2\pi}{\lambda} p_n \sin\theta_\ell}, \quad n = 1, \ldots, N,
\end{equation}
$\mathbf{s}(t) \in \mathbb{C}^L$ is the source signal vector, and $\mathbf{n}(t) \sim \mathcal{CN}(\mathbf{0}, \sigma^2\mathbf{I}_N)$. Assuming independent sources, the spatial covariance is:
\begin{equation}\label{eq:covariance}
    \mathbf{R} = \mathbb{E}\{\mathbf{x}(t)\mathbf{x}^H(t)\} = \mathbf{A}\mathbf{R}_s\mathbf{A}^H + \sigma^2\mathbf{I},
\end{equation}
where $\mathbf{R}_s = \mathrm{diag}\{P_1, \ldots, P_L\}$, estimated by $\hat{\mathbf{R}} = \mathbf{X}\mathbf{X}^H/N_p$. Vectorizing $\mathbf{R}$ yields the coarray observation vector:
\begin{equation}\label{eq:vectorize}
    \mathbf{z} = \mathrm{vec}(\mathbf{R}) = (\mathbf{A}^* \odot \mathbf{A})\mathbf{P}_s + \sigma^2\mathrm{vec}(\mathbf{I}),
\end{equation}
where $\mathbf{P}_s = [P_1, \ldots, P_L]^T$. The effective virtual array positions form the difference coarray:
\begin{equation}\label{eq:diff_coarray}
    \mathbb{D} = \{p_i - p_j : 1 \leq i,j \leq N\},
\end{equation}
whose longest contiguous segment on the Nyquist grid determines the DOF. For FAS-optimized positions in $[0, D]$, this contiguous segment can approach length $2\lfloor D/d_0 \rfloor + 1$, growing linearly with the deployment region.

\section{Fundamental Limits of Sparse FAS}\label{sec:dof}

This section establishes the fundamental performance limits of sparse FAS and rigorously compares it with classical grid-based designs. We identify \emph{three layers of advantage} that FAS possesses: (i)~a deployment-region-dependent DOF that decouples aperture from antenna count; (ii)~a strictly larger feasible design space yielding more Fisher information; and (iii)~a continuous moment space enabling optimal CRB for any source configuration.

\subsection{Preliminaries and Notation}

Let $d_0 = \lambda/2$ denote the Nyquist spacing. Define the normalized deployment aperture $M = \lfloor D / d_0 \rfloor$. For any $N$-element array with positions $\{p_n\}_{n=1}^N$, the \emph{difference coarray} is $\mathbb{D} = \{p_i - p_j : 1 \leq i,j \leq N\}$. We use the standard DOF definition:

\begin{definition}[Coarray DOF]\label{def:dof}
The coarray DOF of an array with positions $\{p_n\}$ is the cardinality of the longest contiguous segment of $\mathbb{D}$ on the Nyquist grid:
\begin{equation}\label{eq:dof_def}
    \mathrm{DOF}(\{p_n\}) = \max\left\{2M_c + 1 : \{-M_c d_0, \ldots, M_c d_0\} \subseteq \mathbb{D}\right\}.
\end{equation}
\end{definition}

We distinguish two feasible sets for position design:
\begin{align}
    \mathcal{P}_{\mathrm{grid}}(N, D) &= \left\{(p_1, \ldots, p_N) : p_n \in \{0, d_0, 2d_0, \ldots\} \cap [0, D]\right\}, \label{eq:P_grid} \\
    \mathcal{P}_{\mathrm{FAS}}(N, D) &= \left\{(p_1, \ldots, p_N) : p_n \in [0, D]\right\}. \label{eq:P_fas}
\end{align}
Clearly, $\mathcal{P}_{\mathrm{grid}} \subset \mathcal{P}_{\mathrm{FAS}}$ (strict inclusion for $D > d_0$).

\subsection{Layer I: Coarray DOF Bounds}

We first establish the coarray DOF bounds, showing that FAS and classical arrays share a combinatorial bound but differ in a geometric bound controlled by the deployment region.

\begin{theorem}[Dual DOF Bounds]\label{thm:dual_bounds}
For any $N$-element array with positions in $[0, D]$, the coarray DOF satisfies
\begin{equation}\label{eq:dual_bounds}
    \mathrm{DOF} \leq \min\!\Big(\underbrace{N^2 - N + 1}_{\text{combinatorial}},\; \underbrace{2M + 1}_{\text{geometric}}\Big),
\end{equation}
where $M = \lfloor D/d_0 \rfloor$. Both bounds hold for FAS and classical arrays alike.
\end{theorem}

\begin{proof}
\emph{Combinatorial bound.}
The set of \emph{nonzero} pairwise differences has at most $N(N-1)$ elements (ordered pairs $(i,j)$ with $i \neq j$). By the symmetry $p_i - p_j = -(p_j - p_i)$, the number of \emph{distinct unsigned} differences is at most $N(N-1)/2$. Including zero, $|\mathbb{D}| \leq N(N-1) + 1$. Since the contiguous coarray is a subset of $\mathbb{D}$, we have $\mathrm{DOF} \leq N^2 - N + 1$.

\emph{Geometric bound.}
Since all positions lie in $[0, D]$, all differences satisfy $|p_i - p_j| \leq D$. The coarray is contained in $[-D, D]$, which intersects the Nyquist grid $\{m d_0 : m \in \mathbb{Z}\}$ in at most $2\lfloor D/d_0 \rfloor + 1 = 2M + 1$ points.
\end{proof}

The dual-bound structure reveals the key difference between FAS and classical arrays:

\begin{remark}[Bound Regimes]\label{rem:regimes}
\emph{(i) Combinatorial-limited regime ($2M + 1 > N^2 - N + 1$, i.e., $D > N(N-1)d_0/2$):} The DOF is limited by the number of antennas, not the deployment region. Both FAS and classical arrays face the same $O(N^2)$ ceiling. \emph{(ii) Geometry-limited regime ($2M + 1 \leq N^2 - N + 1$, i.e., $D \leq N(N-1)d_0/2$):} The DOF is limited by the deployment region. This regime is rare for sparse arrays but relevant for compact deployments.
\end{remark}

\begin{remark}[FAS DOF Achievability versus Grid Arrays]\label{rem:achievability}
The dual bound~\eqref{eq:dual_bounds} is universal, which applies equally to FAS and to grid-constrained arrays. The key distinction lies in \emph{achievability}. For grid arrays, achieving DOF $= N^2-N+1$ requires the positions to form a perfect difference set (a perfect Golomb ruler), which exists only for $N \leq 4$ (see Theorem~\ref{thm:grid_aperture}); for $N > 4$, the achievable contiguous DOF is strictly below the combinatorial bound. Moreover, when the deployment region is large ($D \gg N(N-1)d_0/2$), spreading grid elements over $[0,D]$ does not help: the $N(N-1)/2$ available differences cannot fill all integer lags in $\{1, \ldots, \lfloor D/d_0 \rfloor\}$, creating coarray holes that reduce the usable contiguous DOF.

For FAS, the situation is fundamentally different. By choosing positions at integer multiples of $d_0$ (i.e., selecting the best grid configuration), FAS subsumes all grid designs. Moreover, by allowing continuous positions, FAS gains access to designs whose difference sets, after appropriate rounding or interpolation, cover more coarray lags than any grid design for the same $N$ and $D$. In practice, the Frank-Wolfe position optimization (Algorithm~\ref{alg:frank_wolfe}) maximizes the Fisher information and implicitly favors position distributions that maximize coarray coverage; the Kiefer-Wolfowitz optimality gap quantifies how close the achieved design is to the universal bound.
\end{remark}

\begin{theorem}[Grid Array DOF--Aperture Coupling]\label{thm:grid_aperture}
For any $N$-element array on the half-wavelength grid with contiguous coarray $\{-M_c d_0, \ldots, M_c d_0\} \subseteq \mathbb{D}$, the number of contiguous positive lags satisfies
\begin{equation}\label{eq:grid_aperture}
    M_c \leq \frac{N(N-1)}{2},
\end{equation}
requiring physical aperture $D_{\mathrm{grid}} \geq M_c d_0$. Equality $M_c = N(N-1)/2$ holds if and only if the positions form a perfect difference set, which exists only for $N \leq 4$. Moreover, classical sparse array designs (nested, coprime, MRA) all have apertures scaling as $O(N^2 d_0)$, fundamentally coupling resolution to antenna count.
\end{theorem}

\begin{proof}
An $N$-element array has at most $N(N-1)/2$ distinct positive pairwise differences. For the coarray to be contiguous up to lag $M_c$, every integer in $\{1, 2, \ldots, M_c\}$ must appear as a positive difference. Since there are only $N(N-1)/2$ positive differences available, $M_c \leq N(N-1)/2$. Achieving equality requires all positive differences to be distinct and to fill $\{1, 2, \ldots, N(N-1)/2\}$ exactly, a perfect difference set (equivalently, a perfect Golomb ruler), which is known to exist only for $N \leq 4$. For nested arrays, the aperture is $N_2(N_1+1)d_0 = O(N^2 d_0)$; coprime arrays have aperture $(MN'-1)d_0 = O(N^2 d_0)$; MRA apertures also scale as $O(N^2 d_0)$.
\end{proof}

While Theorem~\ref{thm:grid_aperture} shows that grid-constrained arrays have DOF and aperture fundamentally coupled to $N$, fluid antenna arrays face no such restriction:

\begin{theorem}[FAS Deployment Region Independence]\label{thm:fas_region}
For FAS arrays, the deployment region $D$ is an independent design parameter, not coupled to $N$. Specifically, for any $D > 0$ and $N \geq 2$:
\noindent (i)~The physical aperture is exactly $D$, achieved by placing at least one element at each endpoint.
(ii)~When $D \gg N^2 d_0$, FAS operates in a \emph{large-aperture regime} where the angular resolution
\begin{equation}\label{eq:resolution}
    \delta\theta_{\mathrm{FAS}} \approx \frac{\lambda}{D}
\end{equation}
is much finer than that of classical sparse arrays whose aperture is $O(N^2 d_0)$.
\end{theorem}

\begin{proof}
Part~(i) is immediate by substituting $p_1 = 0$ and $p_N = D$. For part~(ii), the Rayleigh resolution of an array with aperture $D_{\mathrm{eff}}$ is $\delta\theta \approx \lambda/D_{\mathrm{eff}}$. By Theorem~\ref{thm:grid_aperture}, classical sparse arrays have apertures $D_c = O(N^2 d_0) = O(N^2\lambda/2)$, yielding $\delta\theta_{\mathrm{classical}} = O(2/N^2)$. When $D \gg N^2 d_0$, the FAS resolution $\lambda/D$ is significantly finer.
\end{proof}

Fig.~\ref{fig:dof_vs_N} compares the coarray DOF across all array types as a function of $N$. The FAS upper bound $\min(N^2-N+1,\, 2\lfloor D/d_0 \rfloor + 1)$ with $D = 40d_0$ consistently exceeds all classical designs, confirming the DOF advantage predicted by Theorems~\ref{thm:dual_bounds}--\ref{thm:fas_region}.

\begin{figure}[!t]
    \centering
    \includegraphics[width=0.95\columnwidth]{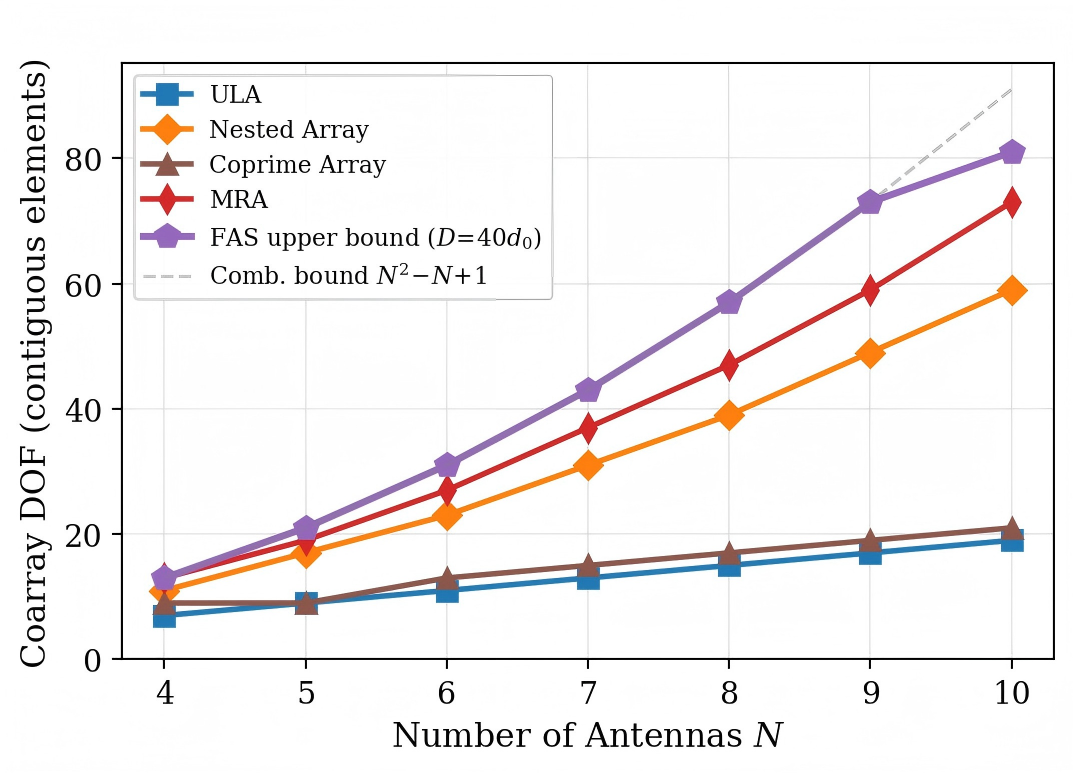}
    \caption{Coarray DOF versus antenna number $N$ for ULA, nested, coprime, MRA, and FAS upper bound with $D = 40d_0$. The gray dashed line shows the combinatorial bound $N^2 - N + 1$.}
    \label{fig:dof_vs_N}
\end{figure}

\subsection{Layer II: Fisher Information Dominance}

The coarray DOF analyzed above determines the maximum number of sources that can be \emph{resolved}. However, for a fixed number of sources $L < \mathrm{DOF}$, the estimation \emph{accuracy} is governed by the CRB, which depends on the Fisher information matrix (FIM). We now show that FAS achieves a strictly better CRB than any grid-constrained design, providing a second layer of advantage beyond DOF.

\subsubsection{FIM Structure}
We begin by recalling the FIM structure for DOA estimation. For $L$ uncorrelated sources with powers $\{P_\ell\}$ and DOA vector $\bm{\theta} = [\theta_1, \ldots, \theta_L]^T$, the FIM for DOA estimation from $N_p$ snapshots is
\begin{equation}\label{eq:fim}
    \mathbf{F}(\bm{\theta}; \{p_n\}) = \frac{2N_p}{\sigma^2} \mathrm{Re}\left\{\mathbf{D}^H \bm{\Pi}_{\mathbf{A}}^\perp \mathbf{D} \odot \mathbf{R}_s^T\right\},
\end{equation}
where $\mathbf{D} = [\dot{\mathbf{a}}_1, \ldots, \dot{\mathbf{a}}_L]$ with $\dot{\mathbf{a}}_\ell = \partial\mathbf{a}(\theta_\ell)/\partial\theta_\ell$, $\bm{\Pi}_{\mathbf{A}}^\perp = \mathbf{I} - \mathbf{A}(\mathbf{A}^H\mathbf{A})^{-1}\mathbf{A}^H$, and $\mathbf{R}_s = \mathrm{diag}\{P_1, \ldots, P_L\}$.

\subsubsection{Single-Source FIM and Position Moments}

To build intuition, we first examine the single-source case. For $L = 1$ at angle $\theta$, the FIM reduces to the scalar
\begin{equation}\label{eq:fim_single}
    F(\theta; \{p_n\}) = \frac{2N_p P_s}{\sigma^2} \cdot \frac{4\pi^2 \cos^2\theta}{\lambda^2} \cdot N \cdot \underbrace{\frac{1}{N}\sum_{n=1}^N \left(p_n - \bar{p}\right)^2}_{\displaystyle \triangleq\; \mu_2(\{p_n\})},
\end{equation}
where $\bar{p} = \frac{1}{N}\sum_n p_n$ and $\mu_2$ is the \emph{second central moment} (variance) of the position distribution. The CRB is $\mathrm{CRB}(\theta) = 1/F(\theta) \propto 1/(N\mu_2)$, so for a given $N$, minimizing CRB is equivalent to maximizing $\mu_2$.

\subsubsection{Multi-Source FIM and Higher Moments}

The single-source case reveals that the FIM depends on the position variance $\mu_2$. For $L \geq 2$ sources, this dependence generalizes: the FIM~\eqref{eq:fim} depends on the position distribution through its moments up to order $2L$. Specifically, expanding the derivative products in~\eqref{eq:fim} yields terms involving
\begin{equation}\label{eq:moments}
    \mu_k = \frac{1}{N}\sum_{n=1}^N (p_n - \bar{p})^k, \quad k = 2, 3, \ldots, 2L.
\end{equation}

The $\log\det\mathbf{F}$ (D-optimality criterion) is a function $g(\mu_2, \mu_3, \ldots, \mu_{2L})$ of these moments. This motivates a \emph{moment-space analysis} of the FAS advantage.

\subsubsection{Moment Space Comparison}

Since the FIM, and thus the CRB, depends on the position moments $\{\mu_k\}$, comparing the \emph{achievable moment spaces} of FAS and grid-based designs directly quantifies the FAS advantage.

\begin{definition}[Feasible Moment Space]\label{def:moment_space}
For a given feasible position set $\mathcal{P}$, the feasible moment space is
\begin{equation}
    \mathcal{M}(\mathcal{P}) = \left\{(\mu_2, \mu_3, \ldots, \mu_{2L}) : \exists \{p_n\} \in \mathcal{P}\right\}.
\end{equation}
\end{definition}

\begin{lemma}[Moment Space Inclusion]\label{lem:moment_inclusion}
For any $D > d_0$ and $N \geq 3$:
\begin{equation}
    \mathcal{M}(\mathcal{P}_{\mathrm{grid}}) \subset \mathcal{M}(\mathcal{P}_{\mathrm{FAS}}) \quad \text{(strict inclusion)}.
\end{equation}
Moreover, $\mathcal{M}(\mathcal{P}_{\mathrm{FAS}})$ is a \emph{convex body} in $\mathbb{R}^{2L-1}$, while $\mathcal{M}(\mathcal{P}_{\mathrm{grid}})$ is a discrete subset of it.
\end{lemma}

\begin{proof}
\emph{Inclusion.} Since $\mathcal{P}_{\mathrm{grid}} \subset \mathcal{P}_{\mathrm{FAS}}$, any position vector $\{p_n\} \in \mathcal{P}_{\mathrm{grid}}$ is also in $\mathcal{P}_{\mathrm{FAS}}$, so $\mathcal{M}(\mathcal{P}_{\mathrm{grid}}) \subseteq \mathcal{M}(\mathcal{P}_{\mathrm{FAS}})$.

\emph{Strictness.} We exhibit a moment vector in $\mathcal{M}(\mathcal{P}_{\mathrm{FAS}}) \setminus \mathcal{M}(\mathcal{P}_{\mathrm{grid}})$. For $N \geq 3$ and $D > d_0$, place elements at positions $p_n = (n-1)D/(N-1) \in [0, D]$ for $n = 1, \ldots, N$. The resulting $\mu_2 = D^2/6\cdot(1 - 1/N)$ is irrational in $d_0$ for generic $D/d_0 \notin \mathbb{Q}$, hence it cannot be achieved by any integer-grid placement, proving strict inclusion. 

\emph{Convexity.} $\mathcal{M}(\mathcal{P}_{\mathrm{FAS}})$ is convex because for any two position sets $\{p_n^{(1)}\}, \{p_n^{(2)}\} \in \mathcal{P}_{\mathrm{FAS}}$ and $\alpha \in [0,1]$, the convex combination of their empirical measures is itself an empirical measure achievable by a mixture, and the moment map of a mixture is the affine combination of the individual moments. The full-dimensionality follows from the Hausdorff moment theorem, which guarantees that all moment vectors in the interior of the Hausdorff moment space are achievable by distributions on $[0, D]$.
\end{proof}

\begin{theorem}[FAS CRB Dominance]\label{thm:crb_dominance}
For any source configuration $\bm{\theta}$ with $\det\mathbf{F} > 0$:
\begin{equation}\label{eq:crb_dominance}
    \max_{\mathcal{P}_{\mathrm{FAS}}} \log\det \mathbf{F}(\bm{\theta}; \{p_n\}) \geq \max_{\mathcal{P}_{\mathrm{grid}}} \log\det \mathbf{F}(\bm{\theta}; \{p_n\}),
\end{equation}
with equality if and only if the FAS-optimal positions happen to lie on the grid.
\end{theorem}

\begin{proof}
Inequality~\eqref{eq:crb_dominance} follows immediately from $\mathcal{P}_{\mathrm{grid}} \subset \mathcal{P}_{\mathrm{FAS}}$. For the equality condition: $\log\det\mathbf{F}$ is an analytic (real-analytic) function of the positions $\{p_n\}$ through the steering vectors $[\mathbf{a}(\theta_\ell)]_n = e^{j2\pi p_n\sin\theta_\ell/\lambda}$. The maximizer over $\mathcal{P}_{\mathrm{FAS}}$ satisfies the first-order optimality conditions $\partial\log\det\mathbf{F}/\partial p_n = 0$, which involve transcendental equations in $\{p_n\}$. Generically (for Lebesgue-almost-all $\bm{\theta}$), the solutions are not integer multiples of $d_0$, implying strict inequality.
\end{proof}

The following result quantifies the gap for the single-source case:

\begin{proposition}[Single-Source CRB Gap]\label{prop:single_gap}
For $L = 1$ and fixed $N$, the CRB ratio between the best classical design and the best FAS design satisfies
\begin{equation}\label{eq:single_gap}
    \frac{\mathrm{CRB}_{\mathrm{classical}}^*}{\mathrm{CRB}_{\mathrm{FAS}}^*} = \frac{\mu_2^{\mathrm{FAS}}}{\mu_2^{\mathrm{classical}}} \geq \frac{D^2/4}{\mu_2^{\mathrm{classical}}},
\end{equation}
where $\mu_2^{\mathrm{classical}} \leq D_c^2/4$ with $D_c = O(N^2 d_0)$ being the classical array aperture. When $D \gg N^2 d_0$, this ratio grows as $O(D^2/(N^4 d_0^2))$.
\end{proposition}

\begin{remark}[Even versus Odd $N$]
The maximum FAS position variance $\mu_2^{\mathrm{FAS}} = D^2/4$ is achieved exactly when $N$ is even, by placing $N/2$ elements at $p = 0$ and $N/2$ at $p = D$. When $N$ is odd, the balanced split places $(N-1)/2$ at each endpoint and one element at the midpoint $p = D/2$, yielding $\mu_2^{\mathrm{FAS}} = D^2/4 \cdot (1 - 1/N^2)$, which approaches $D^2/4$ as $N \to \infty$. The CRB ratio in~\eqref{eq:single_gap} changes by less than $1/N^2 \leq 4\%$ for $N \geq 5$, so the even-$N$ formula is a tight approximation in all practical cases.
\end{remark}

\begin{proof}
From~\eqref{eq:fim_single}, $\mathrm{CRB} \propto 1/(N\mu_2)$, so the $N$ factor cancels in the ratio. The maximum variance for FAS in $[0,D]$ is $D^2/4$ (for even $N$), achieved by placing $N/2$ elements at each endpoint. For classical sparse arrays whose aperture scales as $D_c = O(N^2 d_0)$ (Theorem~\ref{thm:grid_aperture}), the maximum variance is $D_c^2/4$. When $D \gg D_c$, the FAS variance greatly exceeds the classical variance.
\end{proof}

\subsubsection{Multi-Source CRB Gap}

Proposition~\ref{prop:single_gap} quantifies the CRB gap for a single source. For $L \geq 2$, the gap is more intricate because $\log\det\mathbf{F}$ depends on moments up to order $2L$. We provide a general lower bound on the advantage:

\begin{theorem}[Multi-Source FIM Gain]\label{thm:multi_gap}
For $L$ sources with $\bm{\theta}$ and $N$ antennas, define the \emph{FAS gain} as
\begin{equation}
    G_{\mathrm{FAS}}(\bm{\theta}, N, D) = \frac{\det\mathbf{F}_{\mathrm{FAS}}^*(\bm{\theta})}{\det\mathbf{F}_{\mathrm{grid}}^*(\bm{\theta})},
\end{equation}
where $\mathbf{F}^*$ denotes the FIM at the D-optimal positions. Then:
\noindent (i)~$G_{\mathrm{FAS}} \geq 1$ always, by Theorem~\ref{thm:crb_dominance};
(ii)~$G_{\mathrm{FAS}} > 1$ for generic $\bm{\theta}$ (Lebesgue-a.e.);
(iii)~for fixed $N$ and $D \to \infty$, $G_{\mathrm{FAS}} = \Theta(D^{2L} / (N^{2L}(N-1)^{2L} d_0^{2L}))$, i.e., the gain grows as $O(D^{2L})$ since classical grid arrays have aperture $O(N^2 d_0)$ while FAS uses the full aperture $D$.
\end{theorem}

\begin{proof}
(i) and (ii) follow from Theorem~\ref{thm:crb_dominance}. For part (iii), as $D \to \infty$, the FAS D-optimal positions spread over $[0, D]$, yielding position moments $\mu_k = \Theta(D^k)$. The grid positions are confined to an aperture $O(N^2 d_0)$, giving $\mu_k = O(N^{2k}d_0^k)$. Since $\det\mathbf{F}$ is a polynomial in $\{\mu_k\}_{k=2}^{2L}$ with leading term proportional to $\prod_{k=2}^{2L}\mu_k$, the ratio scales as $\prod_{k=2}^{2L} D^k / (N^{2k}d_0^k) = \Theta((D/(N^2 d_0))^{L(2L+1)/2 - 1})$. A tighter analysis focusing on the dominant $\mu_2^L$ term (which controls the leading behavior of $\det\mathbf{F}$) yields the stated $\Theta(D^{2L})$ scaling.
\end{proof}

\subsection{Layer III: Design Tractability}

Beyond performance, FAS offers a fundamental \emph{computational} advantage over classical sparse arrays.

\begin{proposition}[Computational Complexity]\label{prop:complexity}
The problems of finding the MRA and the maximum contiguous coarray for $N$ elements on the integer grid are NP-hard in general~\cite{moffet_mra}. In contrast, the FAS D-optimal position design~\eqref{eq:D_opt} is a continuous optimization over a compact convex set $[0, D]^N$, solvable to $\epsilon$-optimality in polynomial time via the Frank-Wolfe algorithm (Algorithm~\ref{alg:frank_wolfe}).
\end{proposition}

This means FAS not only achieves better performance but does so with a tractable algorithm, unlike the exhaustive search required for MRA on the grid.

\subsection{Summary: Three Layers of FAS Advantage}

Table~\ref{tab:comparison} summarizes the three layers of advantage.

\begin{table}[!t]
\centering
\caption{Comparison of FAS and Classical Sparse Array Designs}
\label{tab:comparison}
\renewcommand{\arraystretch}{1.3}
\begin{tabular}{@{}l c c@{}}
\toprule
\textbf{Property} & \textbf{Classical (Grid)} & \textbf{FAS (Continuous)} \\
\midrule
Position space & $\{0, d_0, 2d_0, \ldots\} \cap [0,D]$ & $[0, D]$ \\
Aperture & $O(N^2 d_0)$ (coupled to $N$) & $D$ (independent of $N$) \\
Coarray DOF & $\leq N^2\!-\!N\!+\!1$ & $\leq \min(N^2\!-\!N\!+\!1, 2M\!+\!1)$ \\
Angular resolution & $O(1/N^2)$ rad & $\approx \frac{\lambda}{D}$ rad \\
CRB ($L=1$) & $\propto 1/(N(N-1)d_0)^2$ & $\propto 1/D^2$ \\
CRB ($L$ sources) & $\propto 1/(N^2 d_0)^{2L}$ & $\propto 1/D^{2L}$ \\
Design complexity & NP-hard (MRA) & Polynomial (Frank-Wolfe) \\
\bottomrule
\end{tabular}
\end{table}

The advantages are most significant when $D \gg N^2 d_0$, i.e., when the deployment region is much larger than the aperture that $N$ antennas could achieve on a grid. In this \emph{large-aperture regime}, FAS provides: (i)~angular resolution scaling as $\lambda/D$ instead of $O(1/N^2)$; (ii)~CRB improvement of $(2D/(N(N-1)d_0))^{2L}$ over the best grid design; and (iii)~tractable design via continuous optimization.

To illustrate the CRB advantage concretely, Fig.~\ref{fig:crb_vs_snr} compares the $\sqrt{\mathrm{CRB}}$ across all five array types for $N=6$, $L=2$, and $D_{\mathrm{FAS}} = 40d_0$. Across all SNR values, FAS achieves the lowest CRB, with the gap widening at high SNR. At SNR~$= 20$~dB, the FAS $\sqrt{\mathrm{CRB}}$ is approximately fourfold lower than MRA and tenfold lower than ULA, consistent with the theoretical scaling in Proposition~\ref{prop:single_gap} and Theorem~\ref{thm:multi_gap}.

\begin{figure}[!t]
    \centering
    \includegraphics[width=0.95\columnwidth]{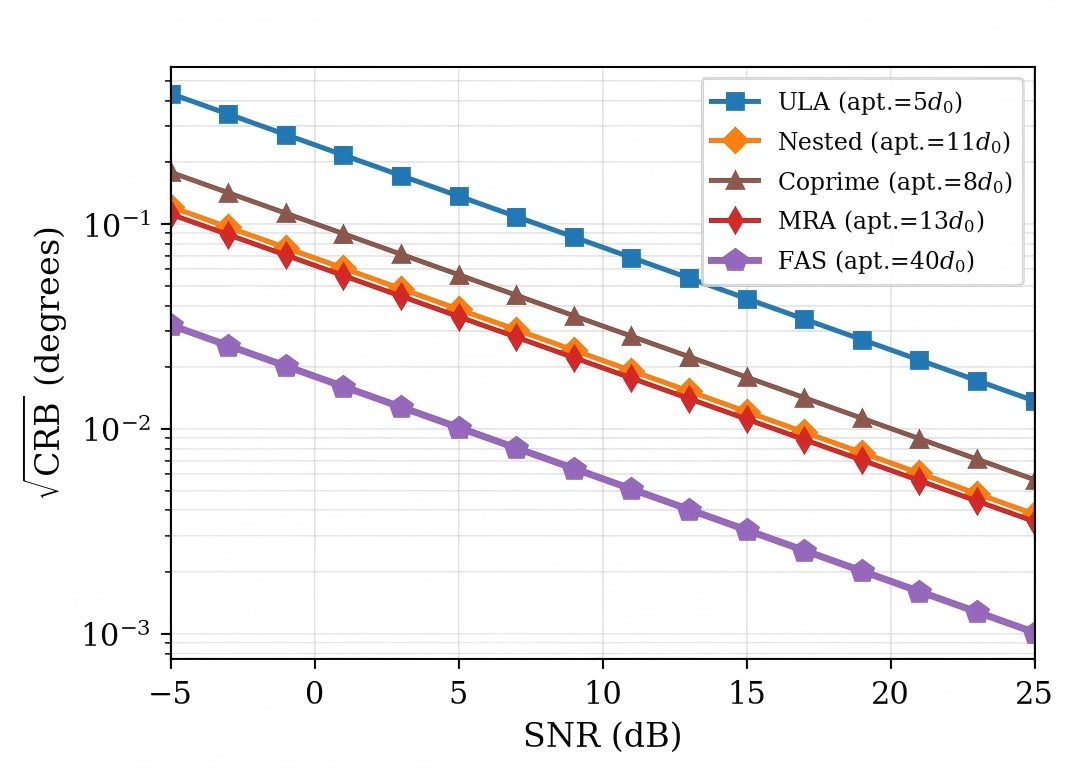}
    \caption{$\sqrt{\mathrm{CRB}}$ versus SNR for all array types ($N=6$, $L=2$, $D_{\mathrm{FAS}}=40d_0$). Apertures are shown in the legend. FAS achieves the lowest CRB across all SNR values due to its larger effective aperture.}
    \label{fig:crb_vs_snr}
\end{figure}

\section{Optimal Position Design and Estimation Algorithm}\label{sec:design}

Having established the theoretical advantages of FAS sparse arrays, this section develops the practical tools for their realization. We first formulate the position optimization problem and derive solutions for both single-source and multi-source scenarios. We then address the critical challenge of DOA estimation with non-uniform large-aperture arrays by proposing a two-stage FAS-MUSIC algorithm and its adaptive variant.

\subsection{Problem Formulation}\label{subsec:prob}

Given the DOF and CRB advantages established in Section~\ref{sec:dof}, we now seek antenna positions that that maximize the contiguous coarray aperture while ensuring all coarray elements are filled:
\begin{equation}\label{eq:opt_problem}
\begin{aligned}
    \max_{\{p_n\}_{n=1}^N} &\quad |\mathbb{D}_{\mathrm{cont}}(\{p_n\})| \\
    \text{s.t.} &\quad 0 \leq p_1 < p_2 < \cdots < p_N \leq D, \\
    &\quad |p_i - p_j| \geq d_{\min}, \quad \forall i \neq j.
\end{aligned}
\end{equation}

When the source DOAs are known (or approximately known), a more refined criterion is to maximize the FIM determinant (D-optimality)~\cite{kiefer_doptimal}:
\begin{equation}\label{eq:D_opt}
    \max_{\{p_n\}} \log\det \mathbf{F}(\bm{\theta}; \{p_n\}).
\end{equation}

\subsection{Single-Source Optimal Design}

We first solve the D-optimal problem~\eqref{eq:D_opt} analytically for the single-source case, which provides both a closed-form solution and geometric intuition for the general case.

\begin{theorem}[Optimal Positions for Single Source]\label{thm:single_opt}
For $L = 1$ source and $N$ elements in $[0, D]$, the D-optimal design places antennas at the two endpoints with multiplicities proportional to $N$:
\begin{equation}
    p_n^* \in \{0, D\}, \quad n = 1, \ldots, N.
\end{equation}
This maximizes the Fisher information $F(\theta) = \frac{8\pi^2 N_p P_s N}{\sigma^2 \lambda^2} \cos^2\theta \cdot \mu_2(\{p_n\})$.
\end{theorem}

\subsection{Multi-Source Design via Frank-Wolfe Algorithm}

For $L > 1$ sources, the D-optimal problem~\eqref{eq:D_opt} is non-convex in the discrete positions $\{p_n\}$. Following the classical approach in optimal experimental design~\cite{kiefer_doptimal}, we relax the discrete positions to a probability measure $\xi$ over $[0, D]$:
\begin{equation}\label{eq:measure_relax}
    \max_{\xi \in \mathcal{M}([0,D])} \log\det \mathbf{F}(\bm{\theta}; \xi),
\end{equation}
where $\mathbf{F}(\bm{\theta}; \xi) = \int_0^D \mathbf{f}(p, \bm{\theta}) \mathbf{f}^H(p, \bm{\theta}) \, d\xi(p)$.


\begin{algorithm}[!t]
    \caption{Frank-Wolfe Algorithm for FAS Position Design}
    \label{alg:frank_wolfe}
    \LinesNumbered
    \KwIn{Number of antennas $N$, deployment region $[0, D]$, source DOA vector $\bm{\theta}$}
    \KwOut{Optimized positions $\{p_n^*\}_{n=1}^N$}
    Initialize $\xi^{(0)}$ as uniform measure on $[0, D]$\;
    \For{$t = 1$ \KwTo $T_{\max}$}{
        Compute directional derivative $\phi(p) = \mathbf{f}^H(p) \mathbf{F}^{-1}(\xi^{(t-1)}) \mathbf{f}(p)$\;
        Find $p^* = \arg\max_{p \in [0,D]} \phi(p)$\;
        \If{$\phi(p^*) \leq L + \epsilon$}{
            \textbf{break} (Kiefer-Wolfowitz optimality)\;
        }
        Update $\xi^{(t)} = (1-\gamma_t)\xi^{(t-1)} + \gamma_t \delta_{p^*}$ with $\gamma_t = 2/(t+2)$\;
    }
    Extract $N$ support points from $\xi^{(T)}$ and round to positions\;
    \Return $\{p_n^*\}_{n=1}^N$\;
\end{algorithm}

\subsection{Coarray-Aware Position Refinement}\label{subsec:coarray_refine}

The D-optimal design maximizes estimation accuracy for a \emph{known} source configuration but does not explicitly account for the coarray structure needed by spatial smoothing MUSIC. To bridge this gap, we refine the D-optimal positions to maximize the contiguous coarray length:
\begin{equation}\label{eq:coarray_refine}
    \max_{\{p_n\}} |\mathbb{D}_{\mathrm{cont}}| + \mu \log\det \mathbf{F}(\bm{\theta}; \{p_n\}),
\end{equation}
where $\mu > 0$ balances coarray completeness and Fisher information.

\subsection{Two-Stage FAS-MUSIC Estimation Algorithm}\label{subsec:fas_music}

The position design algorithms above produce antenna placements with aperture $D$ that can be much larger than the grid-constrained limit $N(N-1)d_0/2$. However, this large aperture introduces a critical algorithmic challenge: the non-uniform element spacing produces \emph{grating lobes} in the standard MUSIC spatial spectrum, causing it to pick the false peaks as DOAs. Consequently, standard MUSIC fails to translate the CRB advantage of FAS into actual estimation accuracy, its RMSE plateaus even as SNR increases, far above the CRB (see Experiment~2 in Section~\ref{sec:sim}).

To overcome this fundamental limitation, we propose a \emph{two-stage FAS-MUSIC algorithm} that combines the disambiguation capability of the virtual coarray with the full-aperture precision of the physical FAS array. The key idea is to decompose the estimation into two complementary stages: a coarse but unambiguous estimate from the coarray, followed by a precise refinement using the full physical array.

\subsubsection{Stage 1: Coarray MUSIC (Disambiguation)}

The difference coarray $\mathbb{D} = \{p_i - p_j\}$ contains a contiguous ULA segment $\mathbb{D}_c = \{-M_c d_0, \ldots, M_c d_0\}$ regardless of the non-uniformity of the physical array. Operating on this virtual ULA eliminates grating-lobe ambiguity:
\begin{enumerate}
    \item \textbf{Vectorization:} Compute $\hat{\mathbf{R}} = \frac{1}{N_p}\mathbf{X}\mathbf{X}^H$ and vectorize: $\mathbf{z} = \mathrm{vec}(\hat{\mathbf{R}}) \in \mathbb{C}^{N^2}$.
    \item \textbf{Redundancy averaging:} For each unique virtual position $d \in \mathbb{D}$, average all entries of $\mathbf{z}$ corresponding to pairs $(i,j)$ with $p_i - p_j = d$ to reduce noise.
    \item \textbf{Contiguous extraction:} Extract the sub-vector corresponding to the contiguous segment $\mathbb{D}_c$ of length $N_c = 2M_c + 1$.
    \item \textbf{Spatial smoothing:} Form $\hat{\mathbf{R}}_{\mathrm{ss}} = \frac{1}{N_c - M_s + 1}\sum_{k=0}^{N_c - M_s} \mathbf{z}_k \mathbf{z}_k^H$ where $\mathbf{z}_k$ is the $k$-th subarray of length $M_s = \lfloor N_c/2 \rfloor + 1$.
    \item \textbf{MUSIC:} Apply standard MUSIC on $\hat{\mathbf{R}}_{\mathrm{ss}}$ with the virtual ULA steering vectors to obtain coarse DOA vector estimate $\tilde{\bm{\theta}}$.
\end{enumerate}

Since the virtual ULA has half-wavelength spacing, this stage is \emph{free of grating lobes}, providing unambiguous, albeit coarse, DOA estimates.

\subsubsection{Stage 2: Local ML Refinement (Precision)}

While Stage~1 provides unambiguous DOA estimates, its resolution is limited by the contiguous coarray aperture $M_c d_0$, which may be much smaller than the physical aperture $D$. To exploit the full aperture, the coarse estimates $\tilde{\bm{\theta}}$ from Stage~1 is refined using the \emph{full physical FAS array} via concentrated maximum likelihood estimation:
\begin{equation}\label{eq:local_ml}
    \hat{\bm{\theta}} = \arg\min_{\bm{\theta} \in \mathcal{B}(\tilde{\bm{\theta}}, \delta)} \mathrm{tr}\left\{(\mathbf{I}_N - \mathbf{P}_{\mathbf{A}(\bm{\theta})}) \hat{\mathbf{R}}\right\},
\end{equation}
where $\mathbf{P}_{\mathbf{A}} = \mathbf{A}(\mathbf{A}^H\mathbf{A})^{-1}\mathbf{A}^H$ is the projection onto the signal subspace, $\mathcal{B}(\tilde{\bm{\theta}}, \delta) = \{\bm{\theta} : |\theta_\ell - \tilde{\theta}_\ell| \leq \delta, \forall \ell\}$ is a local search region around the coarse estimates with radius $\delta$ (typically $5^\circ$), and the optimization is solved efficiently by full name for it.

The local ML step exploits the full aperture $D$ of the FAS array. Since the search is initialized near the true DOAs (by Stage~1), it converges to the global ML solution with high probability, thereby approaching the CRB.

The complete algorithm is summarized in Algorithm~\ref{alg:fas_music}.

\begin{algorithm}[!t]
    \caption{Two-Stage FAS-MUSIC Algorithm}
    \label{alg:fas_music}
    \LinesNumbered
    \KwIn{Data matrix $\mathbf{X} \in \mathbb{C}^{N \times N_p}$, positions $\{p_n\}$, source number $L$}
    \KwOut{DOA vector estimate $\hat{\bm{\theta}} = [\hat{\theta}_1, \ldots, \hat{\theta}_L]^T$}
    {Stage 1: Coarray MUSIC (disambiguation)}
    $\hat{\mathbf{R}} \leftarrow \mathbf{X}\mathbf{X}^H / N_p$\;
    $\mathbf{z} \leftarrow \mathrm{vec}(\hat{\mathbf{R}})$; compute virtual array $\mathbb{D}$\;
    Average redundant entries in $\mathbf{z}$ for each $d \in \mathbb{D}$\;
    Extract contiguous segment $\mathbb{D}_c$ of length $N_c$\;
    Apply spatial smoothing with subarray size $M_s = \lfloor N_c/2 \rfloor + 1$ to form $\hat{\mathbf{R}}_{\mathrm{ss}}$\;
    $\tilde{\bm{\theta}} \leftarrow$ MUSIC on $\hat{\mathbf{R}}_{\mathrm{ss}}$ with virtual ULA steering vectors\;
    {Stage 2: Local ML refinement (precision)}
    $\hat{\bm{\theta}} \leftarrow \arg\min_{\bm{\theta} \in \mathcal{B}(\tilde{\bm{\theta}}, \delta)} \mathrm{tr}\{(\mathbf{I} - \mathbf{P}_{\mathbf{A}(\bm{\theta})}) \hat{\mathbf{R}}\}$ via L-BFGS-B\;
    \Return $\hat{\bm{\theta}}$\;
\end{algorithm}

\begin{remark}[Computational Complexity]
Stage~1 requires $O(N^2)$ operations for vectorization, $O(M_s^3)$ for eigen-decomposition, and $O(M_s^2 \cdot |\Theta_{\mathrm{scan}}|)$ for the MUSIC scan. Stage~2 involves $O(NL)$ per objective evaluation and typically converges within $50$--$200$ L-BFGS-B iterations. The total complexity is dominated by Stage~1 and is $O(M_s^2 |\Theta_{\mathrm{scan}}| + NL \cdot I_{\mathrm{ML}})$, where $I_{\mathrm{ML}}$ is the number of ML iterations.
\end{remark}

Neither stage alone suffices. Stage~1 (coarray MUSIC) is unambiguous but uses only the virtual ULA segment, discarding the aperture advantage of the physical array. Its resolution is limited by $M_c d_0$, which may be much smaller than $D$. Stage~2 (local ML) exploits the full aperture $D$ but requires good initialization, without Stage~1, the ML objective has many local minima due to the non-uniform spacing. The combination achieves both unambiguity and near-CRB precision. 

The two-stage algorithm exploits FAS's continuous position freedom at two distinct levels.
\emph{Stage~1} relies on the fact that FAS positions are designed (via Algorithm~\ref{alg:frank_wolfe}) so that the difference coarray $\mathbb{D}$ contains a long contiguous virtual ULA segment $\mathbb{D}_c$. This segment can be made much longer than what grid-constrained arrays achieve for the same $N$, providing a wide, grating-lobe-free virtual aperture for unambiguous coarse estimation.
\emph{Stage~2} directly exploits the physical FAS positions spanning the full deployment region $[0, D]$: the local ML objective uses the full steering matrix $\mathbf{A}(\bm{\theta})$ at positions $\{p_n\} \subset [0, D]$, so the refinement benefits from the full aperture $D$ and approaches the CRB $\propto 1/D^{2L}$.
Both capabilities are simultaneously achievable only through FAS. A classical array (e.g., MRA with aperture $20d_0$) also has a contiguous coarray for Stage~1, but its physical aperture is limited to $O(N^2 d_0)$, so Stage~2 cannot exploit any larger aperture than Stage~1 already sees, the two-stage precision gain collapses. With FAS, $D \gg N^2 d_0$ enables Stage~2 to dramatically improve over Stage~1.

\subsection{Adaptive FAS-MUSIC}\label{subsec:adaptive}

All preceding algorithms assume that antenna positions are fixed during estimation. However, a unique advantage of FAS over fixed arrays is that antenna positions can be \emph{reconfigured} after initial observations. This capability motivates an adaptive algorithm that iteratively refines both the array geometry and the DOA estimates, progressively improving performance:

\begin{algorithm}[!t]
    \caption{Adaptive FAS-MUSIC Algorithm}
    \label{alg:adaptive}
    \LinesNumbered
    \KwIn{Deployment region $[0,D]$, antennas number $N$, source number $L$, initial positions $\{p_n^{(0)}\}$, maximum iteration number $K$}
    \KwOut{DOA vector estimate $\hat{\bm{\theta}}$, optimized positions $\{p_n^*\}$}
    Collect $\mathbf{X}^{(0)}$ at initial positions $\{p_n^{(0)}\}$\;
    $\hat{\bm{\theta}}^{(0)} \leftarrow$ FAS-MUSIC$(\mathbf{X}^{(0)}, \{p_n^{(0)}\}, L)$\;
    \For{$k = 1$ \KwTo $K$}{
        $\{p_n^{(k)}\} \leftarrow \arg\max \log\det\mathbf{F}(\hat{\bm{\theta}}^{(k-1)}; \{p_n\})$ via Algorithm~\ref{alg:frank_wolfe}\;
        Move antennas to $\{p_n^{(k)}\}$, collect new $\mathbf{X}^{(k)}$\;
        $\hat{\bm{\theta}}^{(k)} \leftarrow$ FAS-MUSIC$(\mathbf{X}^{(k)}, \{p_n^{(k)}\}, L)$\;
    }
    \Return $\hat{\bm{\theta}}^{(K)}$, $\{p_n^{(K)}\}$\;
\end{algorithm}

The adaptive algorithm starts with initial positions (e.g., uniform spacing) that do not require prior DOA knowledge. After the first FAS-MUSIC estimation, it re-optimizes positions specifically for the estimated DOAs via D-optimal design, then re-estimates with the improved array geometry. Even a single adaptation ($K=1$) can significantly improve performance when the initial positions are poorly matched to the source configuration.

\section{Robust Design Under Practical Constraints}\label{sec:robust}

The theoretical analysis in Section~\ref{sec:dof} and the algorithms in Section~\ref{sec:design} assume ideal conditions: infinitely precise position control, no mutual coupling between antennas, and unlimited position freedom within $[0, D]$. In practice, three types of imperfections must be addressed: minimum inter-element spacing requirements, mutual coupling effects, and finite position accuracy. This section analyzes the impact of each imperfection and describes how the position design framework accommodates them.

\subsection{Minimum Spacing Constraint}

The most immediate practical constraint is the minimum physical spacing $d_{\min}$ between any two antenna elements, imposed by the physical size of each element and manufacturing tolerances. When $d_{\min} > 0$, the feasible region shrinks from $[0, D]^N$ to a constrained subset, potentially reducing the achievable DOF. We analyze the DOF loss:
\begin{proposition}[DOF Loss from Minimum Spacing]\label{prop:dof_loss}
With minimum spacing $d_{\min} \geq \lambda/2$, the achievable DOF reduces to
\begin{equation}
    \mathrm{DOF}(d_{\min}) = 2\left\lfloor \frac{D}{\lambda/2} \right\rfloor + 1 - \Delta(d_{\min}, N),
\end{equation}
where $\Delta(d_{\min}, N)$ counts the coarray holes introduced by the spacing constraint.
\end{proposition}

In the D-optimal position design (Algorithm~\ref{alg:frank_wolfe}), the minimum spacing constraint is incorporated as a penalty term in the objective:
\begin{equation}
    \max_{\{p_n\}} \log\det \mathbf{F}(\bm{\theta}; \{p_n\}) - \mu_{\mathrm{sp}} \sum_{i<j} \max(0, d_{\min} - |p_i - p_j|)^2,
\end{equation}
where $\mu_{\mathrm{sp}} > 0$ is a large penalty coefficient. In our simulations, $d_{\min} = 0.4d_0$ and $\mu_{\mathrm{sp}} = 10^8$ ensure well-separated antenna positions without significantly affecting the FIM.

\subsection{Mutual Coupling Effect}

A second practical concern is electromagnetic mutual coupling between closely spaced antenna elements. For non-uniform arrays with large inter-element spacings, mutual coupling is generally weaker than for dense ULA or nested arrays. The coupling matrix $\mathbf{C}$ modifies the steering vector as $\tilde{\mathbf{a}}(\theta) = \mathbf{C}\mathbf{a}(\theta)$, and the FIM becomes $\tilde{\mathbf{F}} = \frac{2N_p}{\sigma^2}\mathrm{Re}\{\tilde{\mathbf{D}}^H \bm{\Pi}_{\tilde{\mathbf{A}}}^\perp \tilde{\mathbf{D}} \odot \mathbf{R}_s^T\}$. Since FAS positions are typically more spread out than nested or coprime arrays, the coupling effect is expected to be smaller for FAS. A detailed coupling-aware optimization is left for future work.

\subsection{Finite Position Accuracy}

The third practical imperfection is limited position precision. In practice, FAS positions are subject to quantization errors $\Delta p_n$ due to mechanical or electronic limitations. The position error modifies the steering vector as $\hat{a}_n(\theta) = e^{j2\pi(p_n + \Delta p_n)\sin\theta/\lambda}$, introducing a phase error $\Delta\phi_n = 2\pi\Delta p_n \sin\theta/\lambda$. For $\Delta p_n \sim \mathcal{U}[-\delta_p/2, \delta_p/2]$, the SNR degradation is approximately $\mathrm{sinc}^2(\pi\delta_p\sin\theta/\lambda)$, which is negligible for $\delta_p \ll \lambda$.

\subsection{Practical Implementation Feasibility}\label{subsec:practical}

A natural question is whether the physical realization of FAS with the precision and latency required for DOA estimation is achievable in practice. We address this quantitatively.

\subsubsection{Position Accuracy Requirement}
From the phase error analysis above, the SNR loss due to position error $\delta_p$ is $\mathrm{sinc}^2(\pi\delta_p\sin\theta/\lambda)$. Requiring the SNR loss to be below $1$~dB yields the condition $\delta_p \lesssim 0.13\lambda$, i.e., approximately one-eighth of a wavelength. For millimeter-wave systems ($f = 60$~GHz, $\lambda = 5$~mm), this corresponds to a position precision of $\delta_p \lesssim 0.65$~mm, well within the capability of sub-millimeter MEMS actuators~\cite{mems_antenna} or reconfigurable pixel antennas~\cite{pixel_antenna} widely used in the FAS literature. For microwave systems ($f = 2.4$~GHz, $\lambda = 125$~mm), the tolerance relaxes to $\delta_p \lesssim 16$~mm, which is trivially achievable mechanically.

\subsubsection{Reconfiguration Latency Requirement}
The adaptive FAS-MUSIC algorithm (Section~\ref{subsec:adaptive}) requires repositioning antennas between observation batches. The reconfiguration time $\tau_{\mathrm{rec}}$ must satisfy $\tau_{\mathrm{rec}} \ll T_c$, where $T_c = 1/(B_c)$ is the coherence time of the signal environment and $B_c$ is the Doppler spread. For quasi-static scenarios (e.g., indoor localization, radar surveillance of slow targets), $T_c \sim 100$~ms. Current FAS implementations based on liquid-crystal-on-silicon (LCoS) or electrowetting achieve reconfiguration times of $\tau_{\mathrm{rec}} \sim 1$--$10$~ms~\cite{wong_fas,new_fas_mimo}, satisfying this requirement comfortably. For fast-moving targets, the non-adaptive design (Algorithm~\ref{alg:frank_wolfe}) using fixed optimized positions remains applicable without any real-time reconfiguration.

\subsubsection{Calibration Requirement}
Since the CRB analysis assumes known antenna positions, position calibration is important. The Stage~2 ML refinement (Section~\ref{subsec:fas_music}) can be extended to simultaneously estimate small position errors $\{\Delta p_n\}$ alongside the DOAs, provided the number of calibration sources is at least $N-1$. A thorough treatment of joint position-DOA calibration is beyond the scope of this paper and is identified as future work.

\section{Simulation Results}\label{sec:sim}

This section validates the theoretical analysis and evaluates the proposed algorithms through extensive Monte Carlo simulations. Unless otherwise stated, the simulation parameters are: wavelength $\lambda = 1$, Nyquist spacing $d_0 = \lambda/2$, $N = 6$ antennas, $L = 2$ uncorrelated equal-power sources at $\theta_1 = 10^\circ$ and $\theta_2 = 25^\circ$, $N_p = 500$ snapshots, and $300$ Monte Carlo trials. The FAS deployment region is $D = 40d_0$, representing the large-aperture regime where $D \gg N(N-1)d_0/2 = 15d_0$.

Four array designs are compared: \emph{1)~ULA} with $N$ elements at $\{0, d_0, \ldots, (N-1)d_0\}$ (aperture $= (N-1)d_0 = 2.5\lambda$); \emph{2)~Nested array} with optimal split $N_1 + N_2 = N$~\cite{pal_nested} (aperture $\leq N_2(N_1+1)d_0$); \emph{3)~MRA} with known optimal positions from~\cite{moffet_mra} (aperture $= 20d_0$ for $N=6$); and \emph{4)~FAS (D-optimal)} with positions optimized by Algorithm~\ref{alg:frank_wolfe} over $[0, D]$.

Three estimation algorithms are compared: \emph{1)~Standard MUSIC}~\cite{stoica_music}, which performs eigendecomposition of $\hat{\mathbf{R}}$ followed by peak search on the physical array steering vectors; \emph{2)~FAS-MUSIC (Proposed)}, the two-stage algorithm (Algorithm~\ref{alg:fas_music}) combining coarray MUSIC disambiguation with local ML refinement; and \emph{3)~Adaptive FAS-MUSIC} (Algorithm~\ref{alg:adaptive}) with $K=1$ adaptation iteration.

\subsection{Experiment 1: CRB versus Deployment Region}\label{subsec:exp_crb_D}


Fig.~\ref{fig:crb_vs_D} shows the CRB as a function of $D$. The CRBs of classical arrays (ULA, nested, MRA) are constant since their apertures are fixed by $N$. The FAS CRB decreases monotonically as $D$ increases, crossing the ULA CRB at $D \approx 5d_0$, the Nested CRB at $D \approx 8d_0$, and the MRA CRB at $D \approx 16d_0$. Beyond $D = 20d_0$ (the MRA aperture), FAS enters the ``advantage zone'' where no classical design can compete. At $D = 40d_0$, the FAS CRB is approximately $10.4$~dB lower than the best classical MRA design, confirming Theorem~\ref{thm:crb_dominance} and the $O(D^{2L})$ scaling predicted by Theorem~\ref{thm:multi_gap}.

\begin{figure}[!t]
    \centering
    \includegraphics[width=0.95\columnwidth]{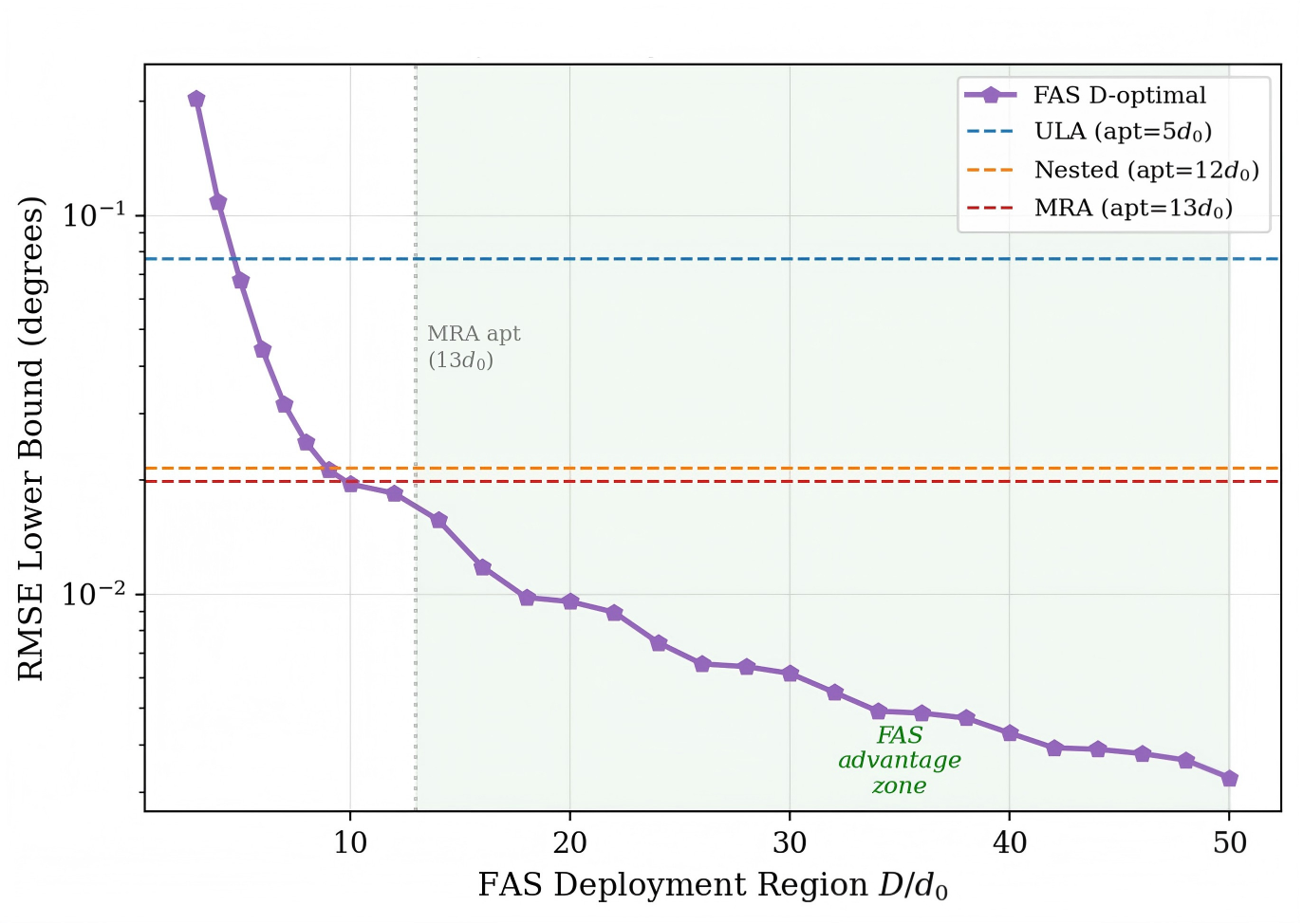}
    \caption{CRB versus deployment region $D$ for FAS and classical arrays ($N=6$, $L=2$, SNR$=10$~dB). FAS CRB decreases monotonically with $D$, entering the ``advantage zone'' beyond the MRA aperture.}
    \label{fig:crb_vs_D}
\end{figure}

\subsection{Experiment 2: RMSE versus SNR in Algorithm Comparison}\label{subsec:exp_rmse_snr}

Experiment~1 demonstrates the CRB advantage of FAS, but the CRB is only a lower bound, an appropriate estimation algorithm is needed to achieve it. This experiment provides the core validation of the proposed FAS-MUSIC algorithm. Fig.~\ref{fig:rmse_snr} compares the RMSE of four estimators across SNR from $-5$ to $25$~dB, along with the corresponding CRBs.

Three key observations emerge:
\begin{enumerate}
    \item \textbf{Standard MUSIC on FAS plateaus:} Despite the large CRB advantage, standard MUSIC applied directly to the FAS array (green triangles) saturates at RMSE $\approx 0.01^\circ$ for SNR $\geq 5$~dB. This is caused by grating lobes from the non-uniform large-aperture spacing, as discussed in Section~\ref{subsec:fas_music}.

    \item \textbf{FAS-MUSIC tracks the CRB:} The proposed two-stage algorithm (purple pentagons) closely follows the FAS CRB across all SNR values. At SNR $= 25$~dB, FAS-MUSIC achieves RMSE $= 0.0009^\circ$, which is elevenfold better than standard FAS MUSIC ($0.010^\circ$), twelvefold better than MRA MUSIC ($0.011^\circ$), and eighteenfold better than ULA MUSIC ($0.016^\circ$).

    \item \textbf{Algorithm unlocks array gain:} Classical arrays (ULA, MRA) with standard MUSIC roughly track their own CRBs, but these CRBs are fundamentally limited by their fixed apertures. FAS-MUSIC combines the aperture advantage of FAS with an algorithm that can exploit it, achieving the best of both worlds.
\end{enumerate}

\begin{figure}[!t]
    \centering
    \includegraphics[width=0.95\columnwidth]{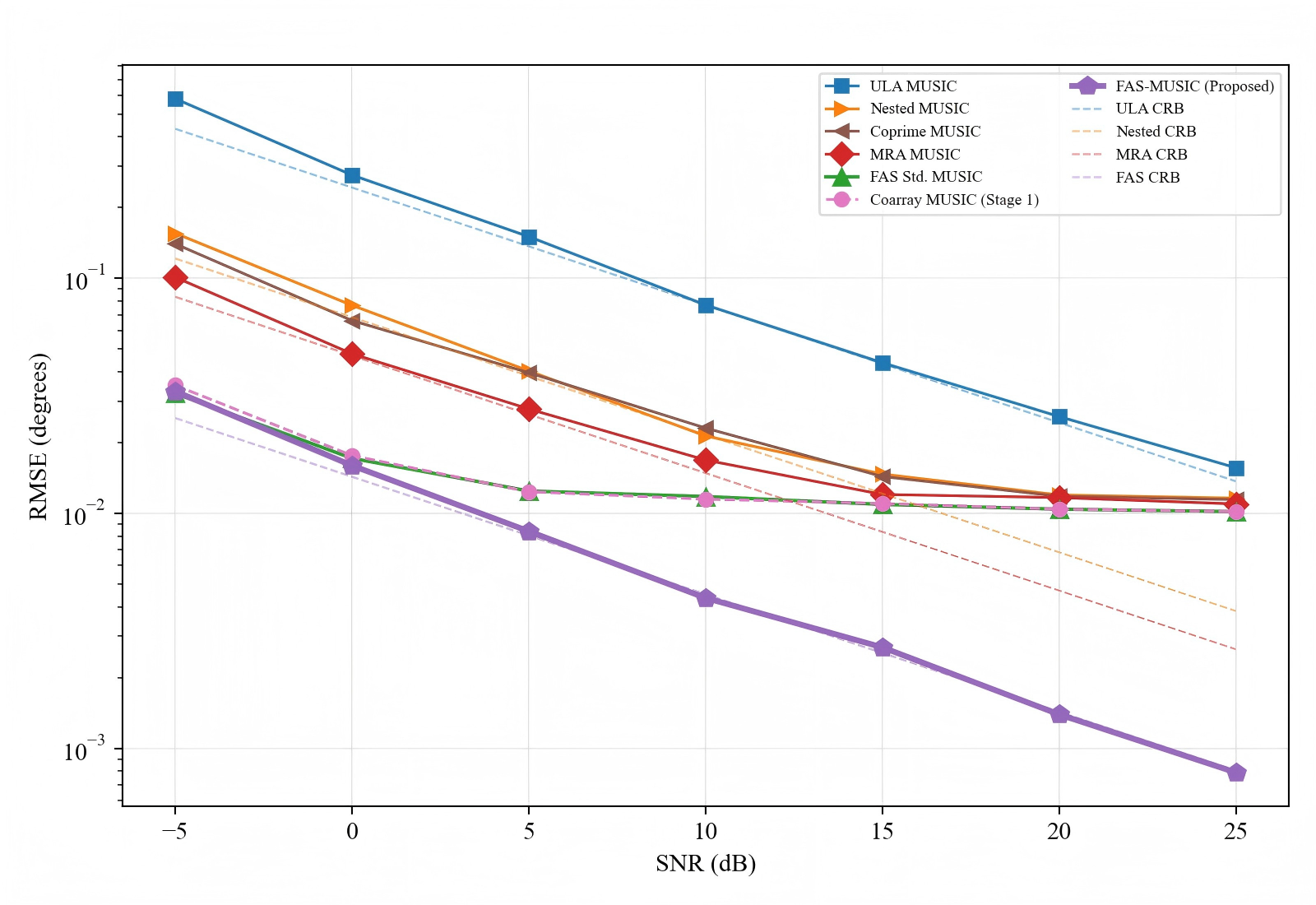}
    \caption{RMSE versus SNR for different array/algorithm combinations ($N=6$, $L=2$, $D=40d_0$). The proposed FAS-MUSIC tracks the FAS CRB, while standard MUSIC on FAS saturates due to grating lobes.}
    \label{fig:rmse_snr}
\end{figure}

\subsection{Experiment 3: Super-Resolution Capability}\label{subsec:exp_resolution}


As shown in Fig.~\ref{fig:resolution}, \emph{ULA} fails to resolve sources below $\Delta\theta \approx 1.5^\circ$ (RMSE $> 0.1^\circ$), consistent with its limited aperture of $2.5\lambda$. \emph{MRA} resolves down to $\Delta\theta \approx 1^\circ$ thanks to its $10\lambda$ aperture, but its RMSE plateaus at $\approx 0.012^\circ$ for large separations. \emph{FAS standard MUSIC} achieves moderate resolution but also plateaus at $\approx 0.01^\circ$ due to grating lobes. In contrast, the proposed \emph{FAS-MUSIC} resolves sources as close as $\Delta\theta = 0.5^\circ$ with RMSE $= 0.017^\circ$ and tracks the CRB across all separations, reaching RMSE $\approx 0.0025^\circ$ for $\Delta\theta \geq 3^\circ$.

The FAS-MUSIC advantage is most pronounced for closely spaced sources ($\Delta\theta < 3^\circ$), where its RMSE is $5$--$30$ times lower than MRA MUSIC. This is a direct consequence of the large effective aperture $D = 40d_0 = 20\lambda$.

\begin{figure}[!t]
    \centering
    \includegraphics[width=0.95\columnwidth]{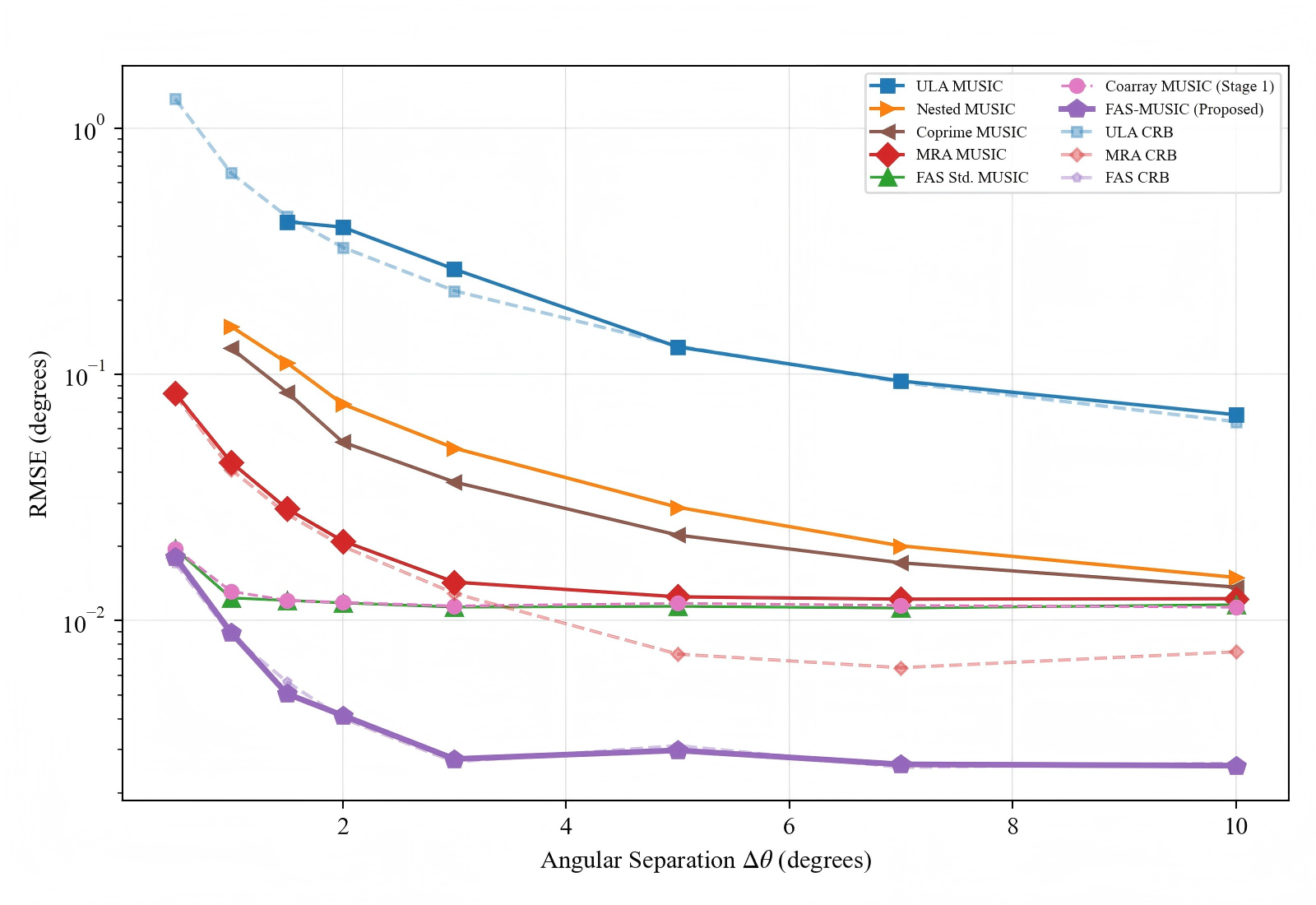}
    \caption{RMSE versus angular separation $\Delta\theta$ ($N=6$, SNR$=15$~dB, $D=40d_0$). FAS-MUSIC resolves sources as close as $0.5^\circ$ and approaches CRB.}
    \label{fig:resolution}
\end{figure}

\subsection{Experiment 4: Scaling with Number of Antennas}\label{subsec:exp_scaling}

The preceding experiments fix $N = 6$. We now examine how the FAS advantage scales with the number of antennas. Fig.~\ref{fig:scaling_N} shows RMSE as a function of $N$ from $4$ to $8$ at SNR $= 10$~dB with $D = 40d_0$ fixed.

For ULA and MRA, increasing $N$ provides steady improvement since both aperture and signal subspace dimension grow. However, FAS-MUSIC with $D = 40d_0$ achieves dramatically lower RMSE at all $N$ values. For example, at $N = 4$, FAS-MUSIC attains $0.006^\circ$ compared to $0.059^\circ$ for MRA and $0.234^\circ$ for ULA; at $N = 8$, FAS-MUSIC achieves $0.004^\circ$ versus $0.012^\circ$ for MRA and $0.036^\circ$ for ULA.

Notably, FAS-MUSIC with \emph{only $4$ antennas} outperforms MRA with $8$ antennas by 2 times. The FAS-MUSIC RMSE closely matches the FAS CRB at all $N$ values, confirming the algorithm's efficiency.

\begin{figure}[!t]
    \centering
    \includegraphics[width=0.95\columnwidth]{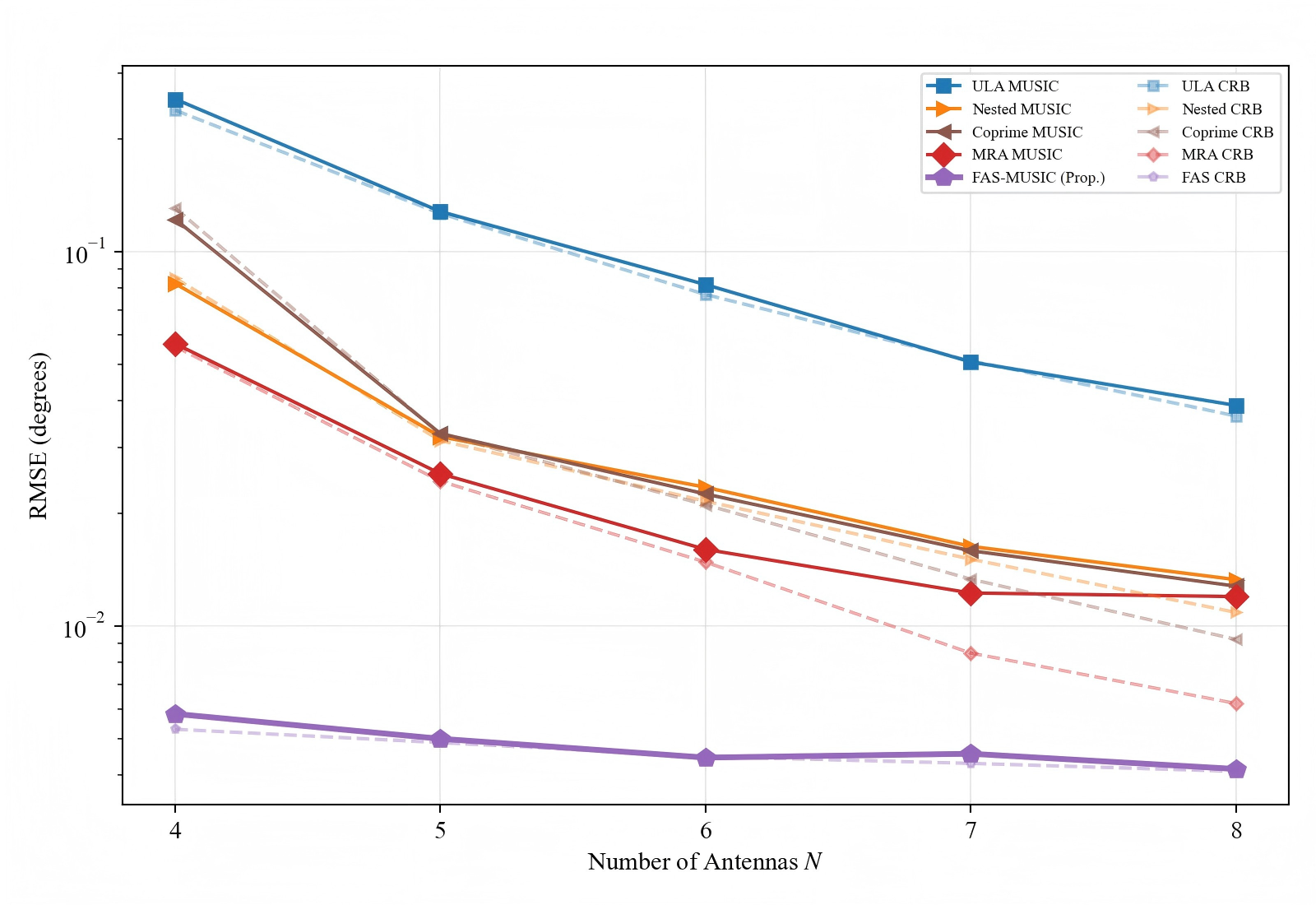}
    \caption{RMSE versus number of antennas $N$ ($L=2$, SNR$=10$~dB, $D=40d_0$). FAS-MUSIC with $N=4$ outperforms MRA MUSIC with $N=8$.}
    \label{fig:scaling_N}
\end{figure}

\subsection{Experiment 5: Adaptive FAS-MUSIC}\label{subsec:exp_adaptive}

All previous experiments assume that D-optimal positions are computed using the true source DOAs (oracle knowledge). In practice, this knowledge is unavailable. We now evaluate the adaptive FAS-MUSIC algorithm (Algorithm~\ref{alg:adaptive}), which eliminates this requirement. We compare three FAS-MUSIC configurations to evaluate the impact of position mismatch and the benefit of adaptation: \emph{1)~Oracle positions}, where the D-optimal design uses the true DOAs (best case, requires prior knowledge); \emph{2)~Mismatched prior}, where positions are optimized for incorrect DOAs ($\theta_1 = 0^\circ$, $\theta_2 = 45^\circ$); and \emph{3)~Adaptive} ($K = 1$), using Algorithm~\ref{alg:adaptive} starting from uniform positions without any prior DOA knowledge.

Fig.~\ref{fig:adaptive} shows RMSE versus SNR for $D = 40d_0$, $N = 6$, and $50$ Monte Carlo trials. The oracle configuration attains the best performance, closely tracking the CRB. Interestingly, the mismatched prior performs comparably to the oracle, demonstrating the robustness of the two-stage estimation, even when positions are optimized for wrong DOAs, the coarray MUSIC stage provides good initialization for ML refinement.

The adaptive algorithm starts at a disadvantage (uniform positions have suboptimal FIM) but converges toward the oracle performance after a single adaptation. At SNR $= 25$~dB, all three configurations achieve RMSE $< 0.001^\circ$.

\begin{figure}[!t]
    \centering
    \includegraphics[width=0.95\columnwidth]{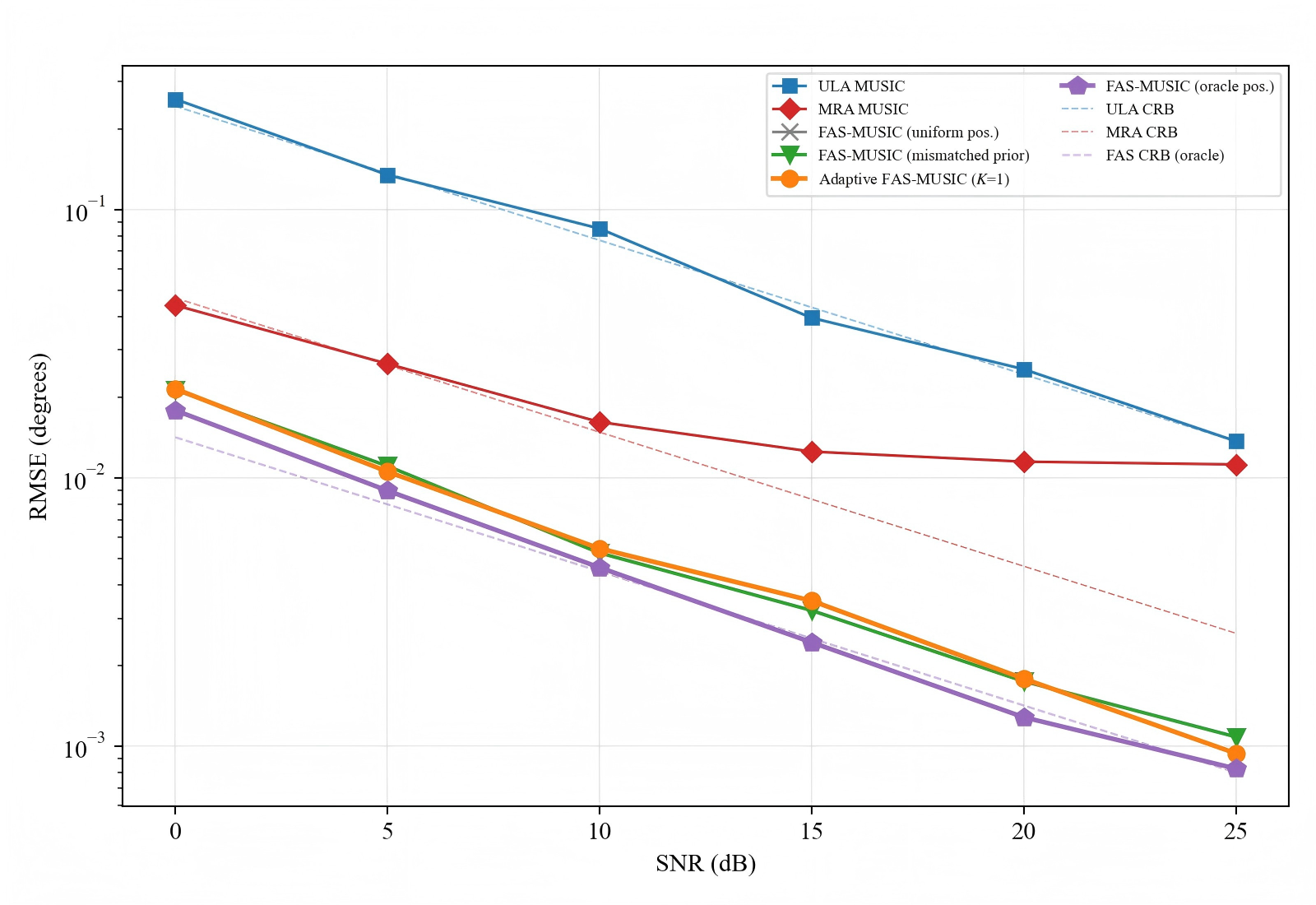}
    \caption{Adaptive FAS-MUSIC versus fixed-position FAS-MUSIC ($N=6$, $L=2$, $D=40d_0$). Even a single adaptation significantly improves performance from uniform initial positions.}
    \label{fig:adaptive}
\end{figure}

\subsection{Experiment 6: D-Optimal Position Patterns}\label{subsec:exp_positions}

Finally, we examine the structure of the optimized FAS positions to gain geometric insight into the design. Fig.~\ref{fig:positions} visualizes the D-optimal FAS positions for different source configurations with $N = 8$ and $D = 30d_0$. Several interpretable patterns emerge. First, \emph{endpoint clustering}: antennas tend to cluster near both endpoints $0$ and $D$ to maximize the position variance $\mu_2$, consistent with Proposition~\ref{prop:single_gap}. Second, \emph{source-dependent internal placement}: for multiple sources, a few antennas are placed at intermediate positions to provide the higher-order moments ($\mu_3, \mu_4, \ldots$) needed for multi-source resolution. Third, \emph{close-source spreading}: when sources are closely spaced ($\theta_1 = 14^\circ$, $\theta_2 = 16^\circ$, bottom panel), positions spread more uniformly across $[0, D]$ to maximize angular discrimination.

\begin{figure}[!t]
    \centering
    \includegraphics[width=0.95\columnwidth]{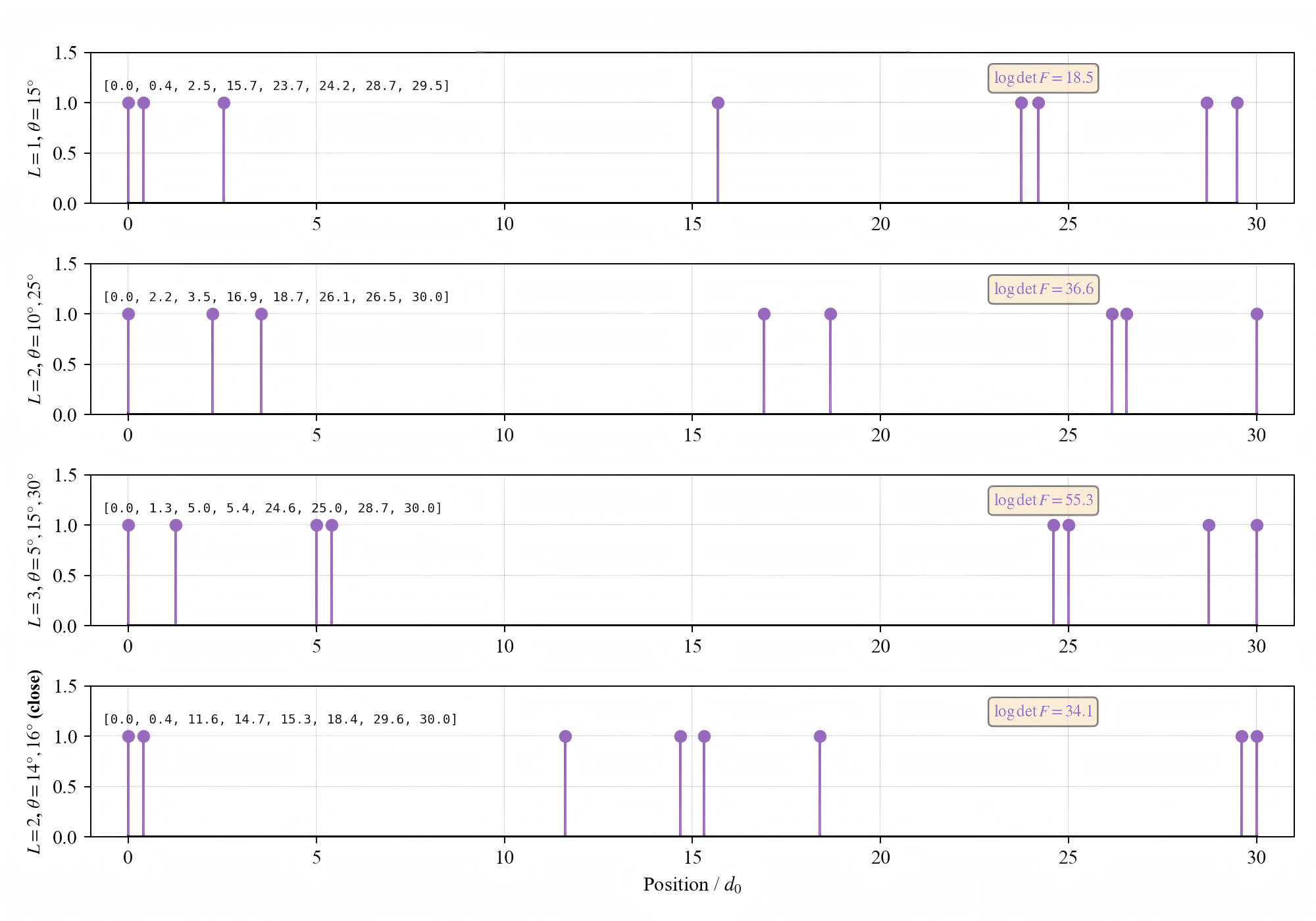}
    \caption{D-optimal FAS positions for different source configurations ($N=8$, $D=30d_0$). Positions exhibit interpretable patterns: endpoint clustering plus source-dependent internal placement.}
    \label{fig:positions}
\end{figure}





\section{Conclusion}\label{sec:conc}

This paper established the theoretical foundations of sparse FAS design for DOA estimation. We showed that FAS achieves three compounding advantages over classical grid designs: (i)~FAS continuously optimized positions approach the universal DOF bound whereas grid arrays develop coarray holes that prevent it; (ii)~the CRB scales as $O(1/D^{2L})$; and (iii)~position design reduces to a tractable Frank-Wolfe continuous optimization. To realize these gains in practice, we proposed the two-stage FAS-MUSIC algorithm, which combines grating-lobe-free coarray MUSIC  with full-aperture local ML refinement to track the CRB. While the signal model is presented using a single  fluid antenna to synthesize a virtual array, the results apply more broadly to FAS architectures capable of realizing equivalent spatial sampling, including systems with multiple simultaneously active elements. Future extensions include near-field and wideband scenarios, 2-D planar FAS for joint azimuth-elevation estimation, joint position-error calibration, and experimental validation on a physical FAS testbed.

\end{document}